\documentclass[final,journal,twoside]{IEEEtran}

\ifCLASSINFOpdf
  \usepackage[pdftex]{graphicx}
  \DeclareGraphicsExtensions{.pdf,.jpeg,.png}
\else
  \usepackage[dvips]{graphicx}
  \DeclareGraphicsExtensions{.eps}
\fi

\usepackage[cmex10]{amsmath}
\usepackage{amssymb,amsthm,amsfonts}
\usepackage{mathabx}
\usepackage{cite}
\usepackage{accents} 
\usepackage[hyphens]{url}
\usepackage{algorithm}
\usepackage{algorithmic}
\usepackage{flushend}
\usepackage{tikz}
\usepackage{pgfplots,filecontents,pgfplotstable}
\usepackage{enumitem}
\usetikzlibrary{calc,positioning,shadows, backgrounds, arrows,shapes,decorations.text,trees}

\usepackage{rotating,amsmath}
\usetikzlibrary{patterns}
\usepgfplotslibrary{fillbetween}
\definecolor{darkblue}{RGB}{0,0,128}
\definecolor{lightgreen}{RGB}{128,255,128}
\definecolor{darkred}{RGB}{128,0,0}
\definecolor{textcol}{RGB}{128,0,0}

\definecolor{darkgreen}{RGB}{0,128,0}
\definecolor{lightyellow}{RGB}{255,255,128}
\definecolor{darkbrown}{RGB}{153,76,0}

\definecolor{lightblue}{RGB}{128,128,255}
\definecolor{boxblue}{RGB}{175,238,238}
\definecolor{lightbrown}{RGB}{204,153,153}
\colorlet{examplefill}{yellow!100!black}

\tikzset{every picture/.style={font issue=\large}, font issue/.style={execute at begin picture={#1\selectfont}}}

\usepackage{xcolor}
\definecolor{darkblue}{RGB}{0,0,128}
\definecolor{darkred}{RGB}{128,0,0}
\definecolor{darkbrown}{rgb}{0.375,0.25,0.125}

\newtheorem{theorem}{Theorem}
\newtheorem{lemma}[theorem]{Lemma}

\theoremstyle{definition}
\newtheorem{definition}{Definition}
\theoremstyle{remark}
\newtheorem{remark}{Remark}

\DeclareMathOperator{\nft}{Q}
\DeclareMathOperator{\inft}{INFT}

\DeclareMathOperator{\shift}{shift}
\newcommand{\nn}{\nonumber}

\newcommand{\vect}[1]{\mathbf{#1}}
\newcommand{\matd}[1]{\mathsf{#1}}
\newcommand{\const}[1]{{\mathcal{#1}}}
\newcommand{\Reals}{\mathbb{R}}
\newcommand{\Realsnn}{\Reals_{0}^{+}}
\newcommand{\Realsnp}{\Reals_{0}^{-}}
\newcommand{\Complex}{\mathbb{C}}

\newcommand{\inner}[2]{\langle #1,#2 \rangle}

\newcommand{\ie}{\emph{i.e.}}
\newcommand{\eg}{\emph{e.g.}}
\newcommand{\der}{\mathrm{d}}
\newcommand{\eqdef}{\stackrel{\Delta}{=}}
\newcommand{\snr}{\text{SNR}}
\providecommand{\norm}[1]{\left\lVert#1\right\rVert}
\newcommand{\E}{\mathsf{E}}
\newcommand{\ft}{\mathcal F}
\newcommand{\naturals}{\mathbb N}
\newcommand{\cev}[1]{\reflectbox{\ensuremath{\vec{\reflectbox{\ensuremath{#1}}}}}}

\newcommand{\normalc}[2]{\mathcal{N}_{\Complex}\!\left(#1,#2\right)}

\newcommand{\period}{\mathcal{T}}
\newcommand{\hadamard}{\circ}

\renewcommand{\angle}{{\rm{Arg}}}

\begin{document}

\thispagestyle{empty}

\title{Linear and Nonlinear Frequency-Division Multiplexing}

\author{Mansoor Yousefi and Xianhe Yangzhang
\thanks{
The material in this paper was presented in part at the 2016 European Conference and Exhibition on Optical Communications.
Mansoor Yousefi is with T\'el\'ecom Paris Tech, 75013 Paris, France (email: yousefi@telecom-paristech.fr). 
Xianhe Yangzhang is with the Department of Electrical and Electronic Engineering, University College London, WC1E 7JE London,
UK (e-mail: x.yangzhang@ucl.ac.uk).}%
}

\markboth{Yousefi and Yangzhang}{Linear and Nonlinear Frequency-Division Multiplexing}

\date{}

\maketitle

\IEEEpeerreviewmaketitle

\begin{abstract}
Two signal multiplexing schemes for optical fiber communication are considered:  
Wavelength-division multiplexing (WDM) and nonlinear 
frequency-division multiplexing (NFDM), based on the nonlinear Fourier
transform (NFT). Achievable information rates (AIRs) of NFDM and WDM are compared 
in a network scenario with an ideal lossless model of the
optical fiber in the defocusing regime. 
It is shown
that the NFDM AIR is greater than the WDM AIR subject to a bandwidth
and average power constraint, in a representative system with
one symbol per user. 
The improvement results from nonlinear signal multiplexing.

\end{abstract}



\section{Introduction}

\IEEEPARstart{O}{ne} factor limiting data rates in optical communication  
is that linear multiplexing is applied to the nonlinear optical fiber. 
To address this limitation, nonlinear
frequency-division multiplexing (NFDM) was introduced \cite{yousefi2012nft1,yousefi2012nft2}, 
\cite[Sec.~II]{yousefi2012nft3}. NFDM is a signal multiplexing scheme based on the nonlinear Fourier 
transform (NFT), which represents a signal in terms of its discrete and continuous
nonlinear Fourier spectra. In NFDM users' signals are multiplexed in the nonlinear Fourier domain and propagate independently in a 
model of the optical fiber described by the lossless noiseless 
nonlinear Schr\"odinger (NLS) equation \cite{yousefi2012nft1,yousefi2012nft2,yousefi2012nft3}.

Prior work illustrates how NFDM is applied, with examples of achievable information rates (AIRs). 
However, an NFDM AIR higher than the
corresponding one in a linear multiplexing has not yet been 
demonstrated. In this paper we consider an optical fiber network with an 
integrable model of the optical fiber in the defocusing regime.  
The main contribution of the paper is showing that the AIR of NFDM 
is greater than the AIR of wavelength-division multiplexing (WDM) for a given 
bandwidth and average signal power, in a representative system with one symbol per user.

AIRs of WDM and NFDM with continuous spectrum modulation 
were compared in \cite[Sec.~V.~D, Fig. 9(b)]{yousefi2012nft3}. 
In this work, although the signal of each user is \emph{modulated} in the 
nonlinear Fourier domain, the signal of different users are \emph{multiplexed} 
linearly. The reason was that the computational complexity of the
inverse NFT, which is usually governed by integral equations, 
is high for data transmission. This made it difficult  to perform both nonlinear
modulation and multiplexing. 
As a consequence, NFDM expectedly achieved data rates approximately equal to WDM rates. 
It was explained that, to improve on WDM rates, one must consider a network 
scenario and multiplex users' signals nonlinearly \cite[Sec. 6.7. 1]{yousefiphdthesis}, 
\cite[Secs. II.C, V. D. and VI. A]{yousefi2012nft3}. 

To demonstrate NFDM high data rates, we first simplify the inverse NFT. 
Common approaches to the inverse NFT are based on integral equations, 
in sharp contrast to approaches to the forward NFT. The integral equations may also 
be cumbersome to solve, hindering 
the application of the NFT. 
In Section~\ref{sec:inverse-nft}, we interpret the inverse NFT as the dual of the 
forward NFT, as in the Fourier transform.
With this perspective, the forward and inverse NFT can be  computed using 
the forward and backward iteration of any integration scheme. This
allows us to 
compute the inverse NFT naturally by running backward the algorithms for the forward NFT described
in \cite{yousefi2012nft2}.
We compare the Boffetta-Osborne (BO) \cite{boffetta1992cds} and the Ablowitz-Ladik (AL) \cite{ablowitz2003dcn} 
integration schemes for the inverse NFT, and apply the AL scheme in Section~\ref{sec:simulations}.

Next we consider the NLS equation in the defocusing regime, which has several advantages.  
First, the operator $L$ in the Lax pair underlying the channel is self-adjoint. Consequently, solitons, the
less tractable part of the NFT, are naturally absent. Second, numerical algorithms are robust in this
regime. 
Third, the analyticity of the one of the nonlinear Fourier coefficients can be exploited to 
compute these coefficients from the NFT efficiently. 
Forth, it is easier to  maximize the spectral efficiency (SE) when the discrete spectrum is absent, as explained below.

In the finite blocklength communication, a guard interval is typically introduced 
between consecutive data blocks in the time domain. 
When the discrete spectrum is absent, 
the NFT is a one-dimensional function similar
to the Fourier transform. Consequently, the Nyquist-Shannon sampling theorem can be applied to
systematically modulate all 
degrees-of-freedoms (DoFs) in a finite nonlinear bandwidth --- namely, in the 
nonlinear Fourier domain, the signal consists of a train of raised-cosine 
functions as in \cite[Eq.~2 \& Sec.~V.~D]{yousefi2012nft3}. 
In contrast, since it is yet unclear how to represent bandlimited $N$-solitons methodically, 
in practice $N$-soliton transmission is limited to small $N$.
As a result, the ratio of the guard time to blocklength is smaller with the continuous spectrum modulation 
than that with the discrete spectrum modulation.


The paper is structured as follows.
The literature on data transmission using the NFT is abundant. In
Section~\ref{sec:related-work} we review this literature, 
highlight the state-of-the-art, and put the present paper in context.
In Section~\ref{sec:networks}  we explain the origin of the limitation of the conventional methods
in optical fiber networks. The forward and inverse NFT are 
presented in Section~\ref{sec:nft-review} and Section~\ref{sec:inverse-nft}, respectively. 
The theoretical base of NFDM was presented in 
\cite{yousefi2012nft1,yousefi2012nft2,yousefi2012nft3}. However, the theory simplifies considerably when
the discrete spectrum is absent. In Section \ref{sec:nfdm}, we present an abstract approach to 
nonlinear modulation and multiplexing using the continuous spectrum. Here, 
in contrast to the linear modulation and multiplexing that are based
on vector space representations, the signal space at the input
of the channel is not represented by a vector space.
AIRs of NFDM and WDM are numerically computed and 
compared in Section~\ref{sec:simulations}. Finally, the paper is concluded in
Section \ref{sec:conclusions}, followed
by Appendices~\ref{sec:A-B} and \ref{sec:kramers-kronig} which contain details.


\subsection*{Notation}

Real, non-negative real, non-positive real, natural and  complex  numbers are denoted by
$\Reals$, $\Realsnn$, $\Realsnp$, $\naturals$ and $\Complex$, respectively. 
The upper half complex plane is the open region 
$\Complex^+\eqdef\bigl\{ \lambda\in\Complex : \Im(\lambda)>0\bigr\}$.
Vectors are distinguished using the bold font, \eg, $\vect{x}\in\Reals^n$.
The Lebesgue space of functions
$f:\Reals\mapsto\Complex$ with finite $p-$norm $\norm{f}_{L^p(\Reals)}$ is represented by
$L^p(\Reals)$. The Hilbert space of
$\period$-periodic complex-valued functions with the inner product $\inner{f}{g}\eqdef\frac{1}{\period}\int_0^{\period}
f(t)g^*(t)\der t$ is shown by $L^2_p([0,\period])$, where $*$ is complex conjugate.
The probability distribution of a (zero-mean) complex circular-symmetric Gaussian random variable with  
variance $\sigma^2$ is denoted by $\normalc{0}{\sigma^2}$.
The Fourier transform operator is $\mathcal F$, defined with the convention in \cite[Eq. 1]{yousefi2014psd}. 
The Fourier transform and the NFT of a signal $q(t)$ are respectively denoted by 
$\mathcal Q(\omega)$, $\omega\eqdef 2\pi f\in\Reals$, and $Q(\lambda)$, $\lambda\in\Complex$.
By default, time, frequency and bandwidth are with respect to the Fourier
transform; the corresponding terms in the NFT are ``nonlinear time'', ``nonlinear frequency'' and 
``nonlinear bandwidth,'' that will be defined formally in Section~\ref{sec:nfdm}. 
Let $\mathcal H$ be a  complex Hilbert space with an orthonormal basis $(\phi_k)_{k\in\mathbb N}$.
We say $N$ is zero-mean Gaussian noise on $\mathcal H$ if $N=\sum_k 
N_k \phi_k$, where $N_k\sim~\normalc{0}{\sigma^2_k}$ is a sequence of independent random variables. 
To simplify the notation, sometimes we add or drop variables in
functions. For example, in Section~\ref{sec:inverse-nft}, $q(t) \eqdef
q(t,z)$ and
$a(\lambda)\eqdef a(\lambda,t)\eqdef a(\lambda, t, z)$.


\section{Transmission Using the NFT}
\label{sec:related-work}

\subsection{NFT in Optical Communication}
\label{sec:nft-optical}

The NFT has appeared in optical communication in assortment of places. 
The literature may be classified as approaches extending linear modulation to 
nonlinear modulation, and linear multiplexing to
nonlinear multiplexing. This division clarifies concepts and explains 
the way that the NFT could usefully be applied to communications.

\subsubsection{Nonlinear modulation}

\paragraph{Multi-soliton communication}
One of the triumphs of the nonlinear science is soliton theory. The NFT is 
the central tool for the analysis of solitons. 
The fundamental soliton communication with on-off keying modulation played a 
pivotal role in the early years of optical communication. 
The fundamental soliton communication was extended to eigenvalue communication in 
\cite{hasegawa1993ec}, in order to transmit more than one 
bit in the time duration of a fundamental 
soliton. The observation was that the eigenvalues are conserved 
in the channel, while the amplitude and phase change. It is thus natural to modulate the 
invariant quantities, such as eigenvalues. 

The 1-soliton communication can be systematically extended to modulating 
all DoFs in $N$-soliton communication \cite{feder2012,yousefi2012nft3,prilepsky2014nis,hari2016,dong2015}.
This extension may be viewed as generalizing the conventional linear modulation in digital
communication \cite[Chap. 6]{gallager2008dcom} to nonlinear modulation
in the optical fiber.
However, although the AIR of $N$-soliton communication is naturally higher than the AIR of 1-soliton 
communication, it is equal to the AIR of any other good signal
set with $N$ parameters. For instance, the widely-used pulse-amplitude modulation (PAM) with Nyquist pulse shape and 
an $N$-ary constellation achieves the same AIR.  
This is because the change of waveform in the channel is not a fundamental limitation in communications, 
since it is equalized at the receiver (as in the radio channels). Importantly, 
the same signal space may be generated by linear or nonlinear modulation.

In sum, nonlinear modulation (such as multi-soliton communication) does not have a
fundamental advantage in terms of the AIR over linear modulation (such as the standard PAM) in any 
channel, under optimal receiver. 
On the contrary, linear modulation with equalization is preferred in  nonlinear 
channels because it is simple.

\paragraph{Discrete and continuous spectrum modulation}
The multi-soliton communication can be generalized to discrete and
continuous spectrum modulation \cite{yousefi2012nft3,prilepsky2014nis,le2014nis,le2015reduced}. 
However, DoFs in time, frequency, nonlinear time, and nonlinear frequency are in
one-to-one correspondence. Thus, modulating both discrete and continuous
spectra does not achieve a data rate higher than the present AIRs 
in coherent systems. In fact, both spectra are indirectly fully modulated in, \eg, Nyquist-WDM
transmission.
Existing approaches are not improved fundamentally by replacing the
set of signal DoFs with another equivalent set.

\subsubsection{Nonlinear multiplexing}
Finally, the NFT can be applied for \emph{signal multiplexing}. For this purpose, 
NFDM was designed in \cite{yousefi2012nft1,yousefi2012nft2,yousefi2012nft3} in order to address the limitation that 
the fiber nonlinearity sets on the AIRs of linear multiplexing in optical networks. 
In this approach, add-drop multiplexers (ADMs) in the network are replaced with 
nonlinear multiplexers constructed using the NFT. Importantly, data rates higher than 
WDM rates may be achieved using this approach.

As noted, modulating both spectra within each user but subsequently 
performing linear multiplexing of users' signals will not achieve a data rate higher than the existing WDM rates. 
On the other hand, linearly modulating the signal of each user 
and performing nonlinear multiplexing of users' signals can improve
WDM rates. NFDM is an approach where modulation can be linear or nonlinear, but
multiplexing is nonlinear. 

Note that although the theoretical principle of the nonlinear multiplexing was presented in 
\cite{yousefi2012nft1,yousefi2012nft2,yousefi2012nft3}, the
simulations in \cite[Sec. V. D]{yousefi2012nft3} are 
essentially nonlinear modulation, combined with linear
multiplexing. It is the objective of the present paper to 
continue the simulations in \cite{yousefi2012nft3} but with nonlinear multiplexing.

\subsection{Review}

Data transmission using the NFT has been explored in numerous papers recently. 
We review some of these papers, highlighting the latest work. 

The NFT consists of a discrete and a continuous spectrum. The discrete spectrum with few nonlinear 
frequencies is modulated in
\cite{buelow2018ptl,buelow2016trans7,yousefi2012nft3,aref2015experimental,aref2016control,
hari2016,terauchi2013,span2017isit,zhang2017cor,garcia2018arxiv}, 
while the continuous spectrum is modulated in  \cite{yousefi2012nft3,prilepsky2014nis,le2014nis}.
The feasibility of the joint discrete and continuous spectrum modulation is 
demonstrated in \cite{aref2016ecoc,tavakkolnia2015sig,aref2018jlt1}.

Noise models for nonlinear frequencies and spectral amplitudes are developed in 
\cite{zhang2015isit,zhang2014sps,civelli2018decision,civelli2018ofc}. 
The probability distribution of the noise in the nonlinear Fourier domain 
is obtained in \cite{moustakas2008nse} for some special cases. The impact 
of noise and perturbations on NFT is studied in 
\cite{civelli2017noise,yangzhang2018jlt,buelow2018ptl,aref2017ecoc1,jones2018ofc}.

Computing the forward NFT is straightforward; several algorithms are 
proposed in \cite{yousefi2012nft2,boffetta1992cds}. 
The inverse NFT finds applications in the fiber Bragg grating design, where in this literature several 
layer-peeling algorithms for the inverse NFT \cite{bruckstein1985} are optimized 
\cite{skaar2003,rosenthal2003,buryak2009ist,civelli2015spawc}. 
Fast NFTs are studied in \cite{wahls2014fnft,wahls2015fast,wahls2015digital,vaibhav2017finft}, 
where here computing the NFT is generally reviewed.

Transmission based on the NFT is extended from single polarization to dual polarization for the 
discrete spectrum in \cite{gaiarin2017ofc,gaiarin2018optica} and for the 
continuous spectrum in \cite{goossens2017optexp}. The periodic NFT
is explored in \cite{kamalian2016} for communications, using 
signals with one or two DoFs for which the inverse NFT can be
computed analytically. 

AIRs of the discrete and continuous spectrum modulation are reported, respectively, in 
\cite{yousefi2012nft3,shevchenko2015,kaza2016isit,hari2016,buelow2018ptl} and 
\cite{yousefi2012nft3,tavakkolnia2017cap,turitsyn2016nature,derevyanko2016cap}.
A lower bound on the 
SE of multi-solitons is obtained  in \cite{kaza2016soliton}. However, these AIRs are obtained under strong
assumptions; among others, the memory is neglected. Furthermore, often these AIRs are obtained 
for point-to-point channels and cannot be compared with the AIRs of the WDM networks. 
Correct analytic calculation of the maximum AIR in the nonlinear Fourier domain is still an 
open problem.

Transmission based on the NFT has been demonstrated in experiments, mostly with low data rates  
\cite{aref2015experimental,buelow2015nonlinear,bulow2014experimental,bulow2015experimental,dong2015,
  maruta2015tiwdc,gui2018cleo}. Aref, B\"ulow and Le conducted a series of 
  experiments showing that one can modulate and demodulate nonlinear frequencies and spectral amplitudes
\cite{buelow2015nonlinear,bulow2014experimental,bulow2015experimental}. 
Recent experiments from this group include joint discrete and continuous spectrum modulation \cite{aref2016ecoc} and 
an experiment modulating 222 nonlinear frequencies in the continuous spectrum reaching 
2.3 bits/s/Hz and 125 Gbps ~\cite{le2017ofc,le2018jlt1}.

Finally, while this paper was under review, 
the main result of the paper that NFDM can outperform WDM 
in simplified models in the defocusing regime has been confirmed in the focusing regime 
\cite{yangzhang2017ofc,yangzhang2018jlt}, dual 
polarization transmission \cite{goossens2017optexp}, and non-ideal models with perturbations 
\cite{yangzhang2018jlt}. 
These related works are based on the methodology established in this paper.
However, a comprehensive comparison of WDM and NFDM in more general and realistic models 
is subject to research.


\section{Optical Fiber Networks}
\label{sec:networks}

In this section, we present the channel model 
and explain the origin of the limitation of the conventional 
communication methods in this model.

\subsection{Network Model}

\begin{figure*}[t]
\centering
\includegraphics[width=0.9\textwidth]{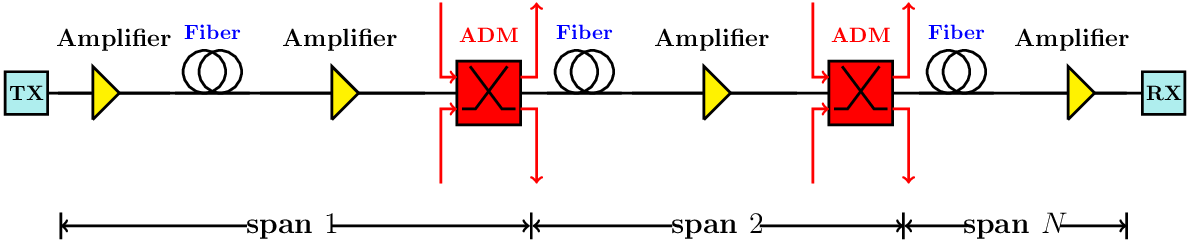}
\caption{Block diagram of a network environment.}
\label{fig:network}
\end{figure*}

The propagation of a signal in the single-mode single-polarization optical fiber 
can be modeled by the stochastic NLS
equation. 
The equation in the normalized form reads \cite[Eqs. 1--3]{yousefi2012nft1}
\begin{IEEEeqnarray}{rCl}
j\frac{\partial q}{\partial z}=\frac{\partial^2 q}{\partial t^2}- 2s|q|^2q+n(t,z),
\label{eq:nls-channel}
\end{IEEEeqnarray}
where $q(t,z):\Reals\times\Realsnn\mapsto\Complex$ is the complex envelope of the signal as a 
function of time $t$ and distance $z$ along the fiber (the transmitter is located at 
$z=0$; the receiver is located at $z=\const L$), $n(t,z)$ is (zero-mean) white circular symmetric 
complex Gaussian noise, and $j\eqdef\sqrt{-1}$. Here, $s\eqdef-1$ in the focusing regime
(corresponding to the standard fiber with anomalous dispersion)  and
$s\eqdef +1$ in the defocusing regime (corresponding to the fiber with normal dispersion).  

The NLS equation \eqref{eq:nls-channel} models the chromatic dispersion (captured by the term $\partial^2 q/\partial t^2$), 
Kerr nonlinearity (captured by the term $|q|^2q$), and noise (that
arises from distributed amplification along the
fiber). We assumed that the fiber 
loss is perfectly compensated by Raman amplification so that \eqref{eq:nls-channel} does
not contain a loss term. The model \eqref{eq:nls-channel} takes into account the
the three leading physical 
effects in fiber; we refer the reader to 
\cite{agrawal2007} for some higher-order effects that
are neglected in \eqref{eq:nls-channel}, and generally for modeling the optical fiber.

In this paper, we consider a \emph{network environment}. This refers
to a communication network 
with the following set of assumptions \cite[Sec.~II. B. 3]{yousefi2012nft3}: (1)
there are  multiple transmitter (TX) and receiver (RX) pairs;
(2) there are add-drop multiplexers (ADMs) in the network. The signal of the user-of-interest can co-propagate
with the signals of the other users in part of the link. The location
and the number of ADMs are unknown; (3) each TX and RX
pair does not know the incoming and outgoing signals in 
the path that connects them. A network environment is depicted in Fig.~\ref{fig:network}.

To make use of the available fiber bandwidth, data can be modulated in disjoint frequency bands. 
As a result of using more bandwidth, data rates, measured in 
bits per second, were rapidly increased with the advent of WDM a few decades ago. 
However, the SE of WDM, measured in bits/s/Hz, is low, vanishing 
with the input power \cite{splett}. From the point of view of communication
theory, WDM is a form of linear multiplexing, a 
fundamental concept forming the basis of the data transmission in most communication systems, including
the fiber-optic systems.

There is a vast literature on WDM AIRs, sometimes referred to as 
the ``nonlinear Shannon limit'' \cite{splett}; see \cite{essiambre2010clo} and references therein.
Fig.~\ref{fig:nft3d} (b) shows the AIR of the WDM as a function of the average 
input power in a network
environment. It can be seen that the AIR vanishes (or saturates in a modified
scheme \cite{agrell2015conds}) as the input power tends to infinity.
The roll-off of the rate with power has been attributed to several factors \cite{tang2001scc}, 
however there is consensus among researchers that the inter-channel interference arising from
nonlinearity is the primary factor
\cite{essiambre2010clo}, \cite[Sec.~II]{yousefi2012nft3}.


\subsection{Origin of the Limitation of the Conventional Methods} 
\label{sec:origin}

It was realized in the past few years that nonlinear interactions do
not limit the capacity in  the deterministic models of optical networks. 
These interactions arise from
\emph{methods of
  communication}, which disregard nonlinearity \cite{yousefi2012nft1,yousefi2012nft2,yousefi2012nft3}.
After abstracting away non-essential
aspects, current methods, in essence, apply linear modulation and multiplexing. 
The linear modulation schemes include PAM and pulse-train transmission. 
The linear multiplexing methods are WDM, orthogonal frequency-division multiplexing (OFDM), 
time-division multiplexing (TDM), polarization-division multiplexing (PDM) and 
space-division multiplexing (SDM).  
When linear multiplexing is applied to nonlinear channels, it gives rise to inter-channel 
interference and inter-symbol interference (ISI). 
In a network environment interference cannot be
removed, while intra-channel ISI can partially be compensated using signal processing. The idea of
sharing bandwidth in WDM, and integration in SDM, conflict with nonlinearity because of  
interactions among transmission modes. Since the nonlinearity is
fixed by physics, it was proposed to 
replace the approach \cite{yousefi2012nft1,yousefi2012nft2,yousefi2012nft3}.  

Yousefi and Kschischang recently proposed nonlinear frequency-division 
multiplexing which is fundamentally compatible with the channel \cite{yousefi2012nft1,yousefi2012nft2,yousefi2012nft3}.  
NFDM exploits a delicate structure in the channel model, in order  
to implement interference-free communication.  NFDM is based on the observation
that the NLS equation supports nonlinear Fourier ``modes'' which have 
an important property that they propagate independently in the
channel, the key to build a multi-user
system. The tool necessary to reveal independent signal DoFs is the NFT. 
Based on the NFT, NFDM was constructed  
which can be viewed as a
generalization of OFDM in linear channels to the nonlinear optical fiber. 
Exploiting the integrability property, NFDM modulates  non-interacting signal DoFs in the channel.


\section{Summary of NFDM}
\label{sec:nft-review}

In this section, we review NFDM from \cite{yousefi2012nft1,yousefi2012nft2,yousefi2012nft3}. 

Let $T: \mathcal H\mapsto\mathcal H$ be a compact (linear) map on a
separable complex Hilbert space $\mathcal H$ with the inner product $<, >$. Consider the channel 
 \begin{IEEEeqnarray}{rCl}
          Y \eqdef T (X)+N,
\label{eq:linear-channel}
\end{IEEEeqnarray}
where $X$ is the input signal, $Y$ is the output signal and $N$ is Gaussian 
noise on $\mathcal H$. 
The channel 
can be discretized by projecting signals and noise onto an orthonormal basis $(\phi_\lambda)_{\lambda\in\naturals}$ of $\mathcal H$
 \begin{IEEEeqnarray}{rCl}
 \bigl\{X, Y, N\bigr\}=\sum\limits_{\lambda=1}^\infty \bigl\{X_\lambda,
 Y_\lambda, N_\lambda\bigr\}\phi_\lambda, 
\end{IEEEeqnarray}
where $X_\lambda, Y_\lambda, N_\lambda\in\Complex$ are DoFs.
This results in a discrete model
 \begin{IEEEeqnarray}{rCl}
   Y_\lambda= H_{\lambda}X_\lambda+
   \underbrace{\sum\limits_{\mu\neq \lambda} H_{\lambda\mu}X_\mu}_{\color{black}\textnormal{linear
       interactions}}+N_\lambda,
\label{eq:lin-isi}
 \end{IEEEeqnarray}
where $H_{\lambda\mu}=\inner{T\phi_\mu}{\phi_\lambda}$, $\lambda\in\naturals$.
Depending on the choice of basis, interactions in \eqref{eq:lin-isi} could refer
to ISI in time, 
inter-channel interference in frequency, or generally interaction
among DoFs in any of the methods in
Fig.~\ref{fig:nft3d}(a). 

Suppose that $T$ is diagonalizable and has a set of eigenvectors forming 
an orthonormal basis of $\mathcal H$, \eg, when $T$ is self-adjoint \cite[Thm. 6]{yousefi2012nft1}. In this basis, interactions in \eqref{eq:lin-isi}
are zero and
 \begin{IEEEeqnarray}{rCl}
   Y_\lambda=H_\lambda N_\lambda+N_\lambda,
   \label{eq:scalar-channels}
 \end{IEEEeqnarray}
where $H_\lambda\eqdef H_{\lambda\lambda}$ is an eigenvalue of $T$. As a result, the channel is 
decomposed into parallel independent scalar channels for  $\lambda=1,2,\cdots$. 

Interactions in \eqref{eq:lin-isi} arise if the basis used for communication 
is not compatible with the channel. As a special case, let $\mathcal
H=L^2_p([0,\period])$ and $T$ be the
convolution map $T(X)\eqdef H(t)\convolution X(t)$, 
where $H(t)\in L^1(\Reals)$ is the channel filter and $\convolution$ denotes
convolution. The eigenvectors and eigenvalues of $T$ are 
\[
\phi_\lambda(t)=\frac{1}{\sqrt\period}\exp(-j \lambda \omega_0 t),\quad \omega_0\eqdef\frac{2\pi}{\period},
\]
and $H_\lambda\eqdef\ft(H)(\lambda\omega_0)$. The Fourier transform maps convolution to a
multiplication operator according to \eqref{eq:scalar-channels}, where
$X_\lambda$, $Y_\lambda$ and $N_\lambda$ are Fourier series coefficients.
Interference and ISI are absent in the Fourier basis. 
OFDM is a technology in which information is modulated in independent spectral 
amplitudes $X_\lambda$, $\lambda\in\mathbb N$.

\begin{figure*}[t]
\centerline{
\begin{tabular}{ccc}
\includegraphics[width=0.25\textwidth]{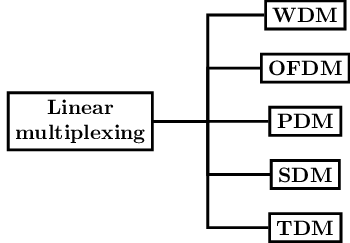}
&

\pgfplotstableread{
SNR AIR
9.2964 3.2875
15.317 4.9158
18.839 5.8781
21.338 6.6271
23.276 7.089
29.296 8.2859
35.317 8.7923
38.839 7.6133
43.276 4.9223
46.798 0.3556
}\wdmcap

\pgfplotstableread{
SNR AIR
9.2964 3.2875
15.317 4.9158
18.839 5.8781
21.338 6.6271
23.276 7.089
29.296 8.2859
35.317 8.7923
50 8.7923
}\wdmcapsaturated


\newcommand{\ptsize}{10pt}
\newcommand{\ptsizee}{8pt}

\tikzset{
  hatch distance/.store in=\hatchdistance,
  hatch distance=8pt,
  hatch thickness/.store in=\hatchthickness,
  hatch thickness=5pt,
  hatch color/.store in=\hatchcolor,
  hatch color=gray!20
}

\pgfdeclarepatternformonly[\hatchdistance,\hatchthickness]{thick vlines}
  {\pgfpointorigin}{\pgfqpoint{\hatchthickness}{100pt}}{\pgfqpoint{\hatchdistance}{100pt}}%
  {
  \pgfsetlinewidth{\hatchthickness}
  \pgfpathmoveto{\pgfqpoint{0.5pt}{0pt}}
  \pgfpathlineto{\pgfqpoint{0.5pt}{100pt}}
  \pgfusepath{stroke}
  }

\pgfdeclarepatterninherentlycolored[\hatchcolor]{crosshatch dots color}
{\pgfpointorigin}{\pgfpoint{\ptsize}{\ptsize}}
{\pgfpoint{\ptsizee}{\ptsizee}}
{
  \pgfsetfillcolor{\hatchcolor}
  \pgfpathrectangle{\pgfpointorigin}{\pgfpoint{8pt}{8pt}}
  \pgfusepath{fill}
  \pgfsetfillcolor{\hatchcolor}
  \pgfpathcircle{\pgfpoint{2pt}{1.75pt}}{20pt}
  \pgfpathcircle{\pgfpoint{6pt}{5.75pt}}{20pt}
  \pgfusepath{fill}
  \pgfsetfillcolor{pgf@darklightsteelblue!70}
  \pgfpathcircle{\pgfpoint{2pt}{2.25pt}}{0.4pt}
  \pgfpathcircle{\pgfpoint{6pt}{6.25pt}}{0.4pt}
  \pgfusepath{fill}
}


\begin{tikzpicture}
\begin{axis}[xmin=0,xmax=51,ymin=2.5,ymax=16,line width=1, font=\normalsize,
legend entries={\textcolor{darkblue}{upper bound}, modified lower bound, \textcolor{darkred}{lower bound}},
legend style={at={(0.715,1)}, draw=none, fill=none, font=\small,legend cell align=left},
xlabel={SNR [dB]},
ylabel={AIR [bits/2D]},
y label style={at={(0.1,0.5)}},
fill=none,
width=0.38\textwidth
]

\addplot[ name path=upperbound, smooth, color=darkblue, line width=1.3] coordinates{
(9.2964, 3.2875)
(50,14.5)
};


\addplot[ name path=lowerbound, color=black, smooth, mark=*, mark size=1pt, line width=1.3] table [y=AIR] {\wdmcapsaturated};


\addplot[color=darkred, smooth, dotted, line width=1.3, restrict x to domain=35:45] table [y=AIR] {\wdmcap};


\addplot[pattern=crosshatch dots color] fill between [of=lowerbound and upperbound];


\node at (210,60) { \color{darkblue}\begin{rotate}{41} $\log(1+\textnormal{SNR})$ \end{rotate} };
\node[draw=none, align=center] at (290,25){\textcolor{textcol}{nonlinearity}\\\textcolor{textcol}{impact}};
\draw[->, thick, densely dotted, color=black] (310,36)--(340,60);

\node at (425,85) {\large \color{black}?};

\end{axis}
\end{tikzpicture}
&
\hspace*{-3mm}
\includegraphics[width=0.4\textwidth]{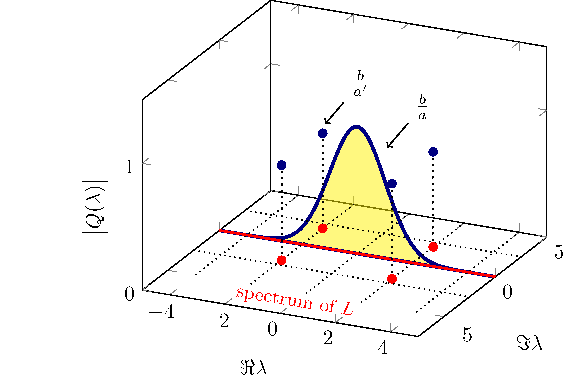} 
\\
~~~~~(a) & ~~~~~~~~(b)
\end{tabular}
}
\caption{(a) Linear multiplexing methods.  (b) Capacity bounds. The lower bound
is the WDM AIR \cite[Fig.~3]{yousefi2012nft3}. (c) The absolute value of the NFT as a surface on the complex
  plane.}
\label{fig:nft3d}
\end{figure*}
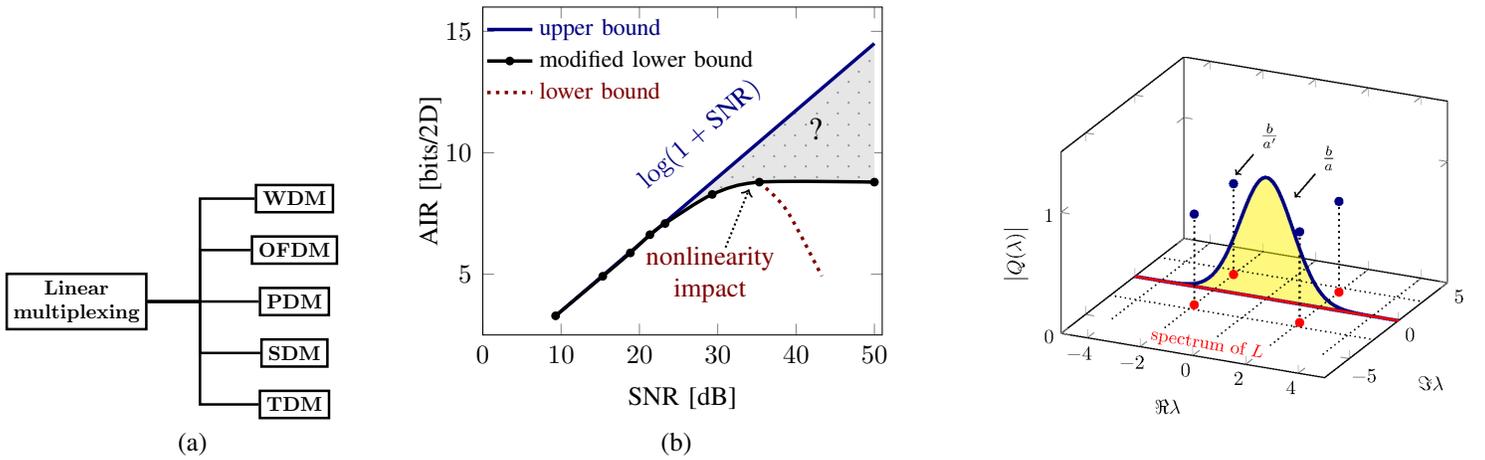

We explain NFDM in analogy with OFDM. First, we define the NFT as
follows. Consider the operator
 \begin{IEEEeqnarray}{rCl}
L\eqdef j\begin{pmatrix}
\dfrac{\partial}{\partial t} & -q(t) \\
sq^*(t) & -\dfrac{\partial}{\partial t}
\end{pmatrix},
\label{eq:L-operator}
\end{IEEEeqnarray}
where $q(t)\in L^1(\Reals)$ is the signal. Let
$\vect{v}(\lambda,t)\eqdef [v_1,v_2]^T$ be an eigenvector of $L$ corresponding to the eigenvalue 
$\lambda$, \ie, 
\begin{IEEEeqnarray}{rCl}
 L\vect{v}=\lambda \vect{v}.
 \label{eq:eigen}
\end{IEEEeqnarray}
The eigenvalues $\lambda$ of $L$ are called \emph{nonlinear frequencies}.
They are complex numbers whose real and imaginary
parts have physical significance \cite[Ex. 1]{yousefi2012nft3}.
Define the normalized eigenvector
\begin{IEEEeqnarray*}{rCl}
\vect{u}(t,\lambda)\eqdef \begin{pmatrix}
a(t,\lambda)\\
b(t,\lambda)
\end{pmatrix},
\end{IEEEeqnarray*}
where
\begin{IEEEeqnarray}{rCl}
a \eqdef e^{j\lambda t}v_1,\quad b\eqdef e^{-j\lambda t}v_2, 
\label{eq:ab-norm}
\end{IEEEeqnarray}
with the initial condition 
$\vect{u}(-\infty,\lambda)=\begin{pmatrix} 1,0\end{pmatrix}^T$.
The nonlinear Fourier coefficients are
$a(\lambda)\eqdef a(\infty,\lambda)$ and $b(\lambda)\eqdef
b(\infty,\lambda)$. The value of the NFT at the nonlinear frequency
$\lambda\in\Reals$, called the \emph{spectral amplitude}, is
$b(\lambda)/a(\lambda)$. 

If $\lambda$ is a simple eigenvalue in the upper half complex
plane $\Complex^+$, it can be shown that $a(\lambda)=0$, so that the first term in the Taylor expansion 
of $a(\zeta)$ around $\zeta=\lambda$
vanishes  \cite[Sec.~IV. B]{yousefi2012nft1}. As a result, the spectral amplitude for 
$\lambda\in\mathbb C^+$ is $b/a'$, where prime denotes
differentiation, because 
$b/a$ is integrated out to its residue $b/a'$ \cite[App.~F]{yousefi2012nft1}. It can be shown that $a(\lambda)$ is an analytic
function of $\lambda$ in $\Complex^+$ \cite[Lem.~4]{yousefi2012nft1}. Thus, nonlinear frequencies in
$\Complex^+$ consist of the discrete set  of (simple) zeros of 
$a(\lambda)$ denoted by $(\lambda_i)_{i\in\mathcal N}$, where $\mathcal N\eqdef\{1,\cdots, N\}$, $N\in\mathbb N$. If $s=1$, the operator $L$ is self-adjoint; 
thus $\mathcal N$ is empty and nonlinear frequencies are real. 

\begin{definition}[Nonlinear Fourier Transform]
The NFT of $q(t)\in L^1(\Reals)$ with respect to the $L$ operator \eqref{eq:L-operator} is a function $Q(\lambda)$ of the 
complex frequency $\lambda\in\Complex$, defined as
  \begin{IEEEeqnarray*}{rCl}
    Q(\lambda) \eqdef
\begin{cases}
\dfrac{b(\lambda)}{a(\lambda)}, & \lambda\in\Reals,\\[1em]
\dfrac{b(\lambda_i)}{a'(\lambda_i)},& \lambda_i\in\Complex^+,\quad
i\in\mathcal N.\\
\end{cases}
  \end{IEEEeqnarray*}
The functions $\hat q(\lambda) \eqdef Q(\lambda)$, 
$\lambda\in\Reals$, and  $\tilde q(\lambda_i) \eqdef Q(\lambda_i)$, $\lambda_i\in\Complex^+$,  
are called, respectively, the continuous and discrete spectrum. 

\qed
\end{definition}
 
Let $Q(\lambda,z)$ be the NFT of $q(t,z)$ with respect to $t$. 
The important property of the NFT is that, if $q(t,z)$ propagates in \eqref{eq:nls-channel}  
with noise set to zero,  we have
\begin{IEEEeqnarray}{rCl}
Q(\lambda,\const L)=H(\lambda,\const L)Q(\lambda,0),
\label{eq:channel-filter}
\end{IEEEeqnarray}
where  $H(\lambda,\mathcal L) \eqdef \exp(j4s\lambda^2 \const L)$ is the all-pass-like channel filter. 
It follows that, just as the Fourier transform converts a linear convolutional channel
into a number of parallel independent channels in frequency, the NFT converts the nonlinear dispersive channel \eqref{eq:nls-channel}
in the absence of noise into a number of parallel independent channels in nonlinear frequency. 
In NFDM information is modulated in independent spectral amplitudes
$\nft(\lambda)$ for every $\lambda$. 
In the absence of noise,
interference and ISI are simultaneously zero for all users of a multiuser network.
Therefore, in contrast to WDM, the NFDM AIR is infinite in the 
deterministic model at any non-zero power.


\section{Computing the Inverse NFT}
\label{sec:inverse-nft}

The standard approaches to the inverse NFT are based on the Riemann-Hilbert or
Gelfand-Levitan-Marchenko integral equations \cite[Ch.~2.2]{ablowitz2003dcn}.
Naive numerical solution of these integral equations may be time-consuming or prone to error; 
optimized implementation requires attention to details of the numerical solution of the integral equations, 
which may produce a digression from the main purpose of using the NFT here. 
Consequently, we seek simpler, more natural and interpretable methods.

We interpret the inverse NFT as the dual of the forward NFT, as in the
Fourier transform. We show that the  forward
and inverse NFT can be computed by running the iterations of any integration scheme, forward
and backward in time.
This allows us to use algorithms for the forward NFT for computing the inverse NFT, 
for example the Boffetta-Osborne and Ablowitz-Ladik 
schemes \cite{yousefi2012nft2}. 

We emphasize that the two algorithms that are presented in this section are not new. 
They exists in applied mathematics \cite{bruckstein1985}, \cite{boffetta1992cds}, are used in the literature of the fiber Bragg gratings 
design \cite{skaar2003,rosenthal2003}, and revisited recently for fast implementation \cite{wahls2015fast,vaibhav2017finft}. 
However, their derivation in the literature has been made unnecessarily over-elaborate, as well as 
intertwined with the details of the physical application. The contribution of this section is to 
re-derive these existing algorithms using essentially a few lines 
of elementary analysis, in a simple and clear manner --- compare 
the equation \eqref{eq:qk-from-vk+1-ALL} with \eg, paper \cite{skaar2003} or \cite{bruckstein1985}, \cite{boffetta1992cds}. 
Section~\ref{sec:bo} and \ref{sec:al} prior to \eqref{eq:qk-from-vk+1-ALL} adapt the forward transform from \cite{yousefi2012nft2}.

The inverse NFT may be divided into three steps. First, from the NFT $Q(\lambda)$ we  obtain 
$a(\lambda)$ and $b(\lambda)$. Second, from $a(\lambda)$ and $b(\lambda)$ we obtain $a(\lambda,t)$ and $b(\lambda,t)$.
Third, we get $q(t)$ from $a(\lambda,t)$ and $b(\lambda,t)$. 

\begin{remark}
The NFT in \cite{yousefi2012nft1} was presented for the focusing regime. 
The equations in \cite{yousefi2012nft1} can often be extended from the focusing regime  
to the general case by substitutions 
\begin{IEEEeqnarray*}{rClrCl}
\{a,a^*\}&\rightarrow& \{a,a^*\}, \qquad\{q,q^*\}&\rightarrow&~\{q, -sq^*\},\\
\{b,b^*\}&\rightarrow& \{b,-sb^*\},\qquad \{\hat q,\hat q^*\}&~\rightarrow& \{\hat
q,-s\hat q^*\}.
\end{IEEEeqnarray*}
For example, the Parseval's identity for the NFT is \cite[Sec. IV. D. 8]{yousefi2012nft1}
\begin{IEEEeqnarray}{rCl}
\int\limits_{-\infty}^{\infty} |q(t)|^2\der t=-\frac{s}{\pi}\int\limits_{-\infty}^\infty
\log\left(1-s|\hat q(\lambda)|^2\right)\der\lambda.
\label{eq:parseval}
\end{IEEEeqnarray}
This implies that $|\hat q(\lambda)|<1$ in the defocusing regime.

In this section, we drop the variable $z$ denoting the distance. Thus,
$q(t)\eqdef q(t,z)$, $a(\lambda)\eqdef a(\lambda,z)$, etc.

\qed
\end{remark}

\subsection{Inverse Boffetta-Osborne Scheme}
\label{sec:bo}

We shall begin with the Boffetta-Osborne integration
scheme for the Zakharov-Shabat system, which is known to perform well for 
the forward NFT \cite{boffetta1992cds,yousefi2012nft2}. 
This scheme is also called the continuous-time layer-peeling (CLP). 

In the BO scheme, $q(t)$ is approximated by a piece-wise constant function. 
The NFT of a constant, and by induction a piece-wise constant, function 
can be computed analytically \cite[Sec.~IV.~C]{yousefi2012nft1}. Let $q(t)$ be a piece-wise constant function supported on 
$[T_1,T_2]$. Discretize the time on the mesh 
\begin{IEEEeqnarray}{rCl}
t[k] \eqdef T_1+k\epsilon, \quad k=0, \cdots, N-1,
\IEEEeqnarraynumspace
\label{eq:t[k]}
\end{IEEEeqnarray}
where $\epsilon\eqdef T/N$, $T\eqdef T_2-T_1$, and set $q[k]\eqdef q(t[k])$. Recall that the forward iteration in the 
BO scheme is \cite[Sec.~III.~C]{yousefi2012nft2}.
\begin{IEEEeqnarray}{rCl}
\vect{u}[k+1,\lambda]=\matd M[k,\lambda, q] \vect{u}[k,\lambda],\quad \vect{u}[0,\lambda]=\begin{pmatrix} 1\\0 \end{pmatrix},
\label{eq:forward-bo}
\end{IEEEeqnarray}
where the \emph{monodromy matrix} $\matd{M}$ is
\begin{IEEEeqnarray*}{rCl}
\matd{M}[k,\lambda,q]\eqdef
\begin{pmatrix}
x[k] & \bar y[k]\\
y[k] & \bar x[k]
\end{pmatrix},
\end{IEEEeqnarray*}
in which \cite[Eq. 11]{yousefi2012nft2}:
\begin{IEEEeqnarray}{rCl}
x[k] &\eqdef&\Bigl(\cos(D\epsilon )-j\frac{\lambda}{D}\sin(D\epsilon )\Bigr)e^{j\lambda \left(t[k]-t[k-1]\right)}, 
\label{eq:xk}\\
y[k]&\eqdef&\frac{sq^*[k]}{D}\sin(D \epsilon)e^{-j\lambda(t[k]+t[k-1])},
\label{eq:yk}
\end{IEEEeqnarray}
where $D\eqdef\sqrt{\lambda^2-s|q[k]|^2}$ and
\begin{IEEEeqnarray*}{rCl}
\bar{x}[k]\eqdef x^*[k](\lambda^*), \quad
\bar y[k]\eqdef sy^*(\lambda^*).
\end{IEEEeqnarray*}
It can be verified that $\det \matd{M}=1$.

The three steps of the inverse NFT via the BO scheme are as follows. 

\subsubsection*{Step 1 (Factorization)}

The first step is obtaining two parameters $a(\lambda)$ and $b(\lambda)$ from one parameter 
$Q(\lambda)$. Fundamentally, this is a Riemann-Hilbert factorization problem in 
the complex analysis \cite[Ch.~]{ablowitz2003dcn}. 
Given $Q(\lambda)$, the Riemann-Hilbert integral equations in
\cite[Eq. 30]{yousefi2012nft1} can be solved at $t=T$ to obtain 
$\vect{V}^2(T,\lambda)\eqdef \vect{u}=[a(\lambda),b(\lambda)]$. This requires solving a
system of linear equations. This step does not incur a high computational cost, because the equations are solved only
at one point $t=T$. 
Furthermore, in the defocusing regime, this step can be done more efficiently as follows.

From the unimodularity condition
\begin{IEEEeqnarray}{rCl}
\left|a(\lambda)\right|^2-s\left|b(\lambda)\right|^2=1, \quad \lambda\in\Reals,
\label{eq:C-unimodularity}
\end{IEEEeqnarray}
we obtain
\begin{IEEEeqnarray}{rCl}
\left|a(\lambda)\right|=\frac{1}{\sqrt{1-s|\hat q(\lambda)|^2}}.
\label{eq:|a|}
\end{IEEEeqnarray}

Since $q(t)\in L^1(\Reals)$,  $a(\lambda)$ can be analytically extended  to 
$\Complex^+$ \cite[Lemma~2.1]{ablowitz2003dcn}. The real and imaginary parts 
of an analytic function are Hilbert transforms of one another. 
In Appendix~\ref{sec:kramers-kronig} it is shown that in the defocusing regime 
$\log a(\lambda)$ can also be analytically extended to a region near the real line, 
thus amplitude $|a(\lambda)|$ and phase $\angle(a(\lambda))$ are related by the Hilbert transform, \ie,
\begin{IEEEeqnarray}{rCl}
 \angle(a(\lambda))=\mathcal H\left(\log \left|a(\lambda)\right|\right),
\quad \lambda\in\Reals,
\label{eq:angle(a)}
\end{IEEEeqnarray}
where $\mathcal H$ denotes the Hilbert transform and $\angle$ is the principal value of the phase. 
As a result, we easily obtain $a$ and $b=\hat q a$ 
in the defocusing regime.

\subsubsection*{Step 2 (Integration)}

The second step is to obtain $a(t,\lambda)$ and $b(t,\lambda)$ from 
$a(\lambda)$ and $b(\lambda)$. This step can be performed by integrating the
Zakharov-Shabat system  in negative time, \ie,  by running the iterations for the forward NFT 
backward in time. 

From \eqref{eq:forward-bo}, the backward iteration is:
 \begin{IEEEeqnarray}{rCl}
\vect{u}[k,\lambda]=\matd{M}^{-1}[k,\lambda, q] \vect{u}[k+1,\lambda],\quad
\vect{u}[N,\lambda]=\begin{pmatrix} a(\lambda)\\b(\lambda) \end{pmatrix},
\IEEEeqnarraynumspace
\label{eq:inverse-ab}
\end{IEEEeqnarray}
where $k=N-1,N-2,\cdots,0$, and 
\begin{IEEEeqnarray*}{rCl}
\matd{M}^{-1}[k,\lambda, q]
=
\begin{pmatrix}
\bar x[k] & -\bar y[k]\\
-y[k] & x[k]
\end{pmatrix}.
\end{IEEEeqnarray*}

Computing $x[k]$ and $y[k]$ in  \eqref{eq:xk} and \eqref{eq:yk} may involve evaluating  
$\cosh(x)$ for some $x$, when $s=1$. Moderate values of $x$, \eg, $x>25$,  result in 
large numbers and numerical error.  The numerical error may be reduced if the ratio 
$\hat q=b/a$ is updated, so that large numbers are canceled between $a$ and $b$. The forward iteration 
for $\hat q[k,\lambda]$ is
\begin{IEEEeqnarray*}{rCl}
\hat q[k,\lambda]&=&\alpha[k]
\frac{\bar\beta[k]+ \hat q[k-1,\lambda]}{1+\beta[k]\hat q[k-1,\lambda]},\quad\hat q[0,\lambda]=0,
\end{IEEEeqnarray*}
where
\begin{IEEEeqnarray*}{rCl}
\alpha[k]&\eqdef&e^{-2j\lambda(t[k]-t[k-1])}\frac{1+\frac{j\lambda}{D}\tan(D\epsilon)}{1-\frac{j\lambda}{D}\tan(D\epsilon)},\\
\beta[k]&\eqdef&e^{2j\lambda t[k-1]}\frac{q[k]}{D}\frac{\tan(D\epsilon)}{1-\frac{j\lambda}{D}\tan(D\epsilon)},
\end{IEEEeqnarray*}
and $\bar\beta[k]\eqdef s\beta^*[k](\lambda^*)$.
The backward iteration is
\begin{IEEEeqnarray*}{rCl}
\hat q[k-1,\lambda]=\frac{\hat q[k,\lambda]-\bar\beta[k]\alpha[k]}{\alpha[k]-\beta[k]\hat q[k,\lambda]},\quad \hat q[N,\lambda]=\hat q(\lambda).
\end{IEEEeqnarray*}

\subsubsection*{Step 3 (Signal Recovery)}

From \eqref{eq:eigen}, $q(t)$ can be read off as
\begin{IEEEeqnarray}{rCl}
q^*(t)&=&s\frac{\partial_t v_2-j\lambda v_2}{v_1}\nn\\
&=&se^{j2\lambda t}\frac{\partial_t b(t,\lambda)}{a(t,\lambda)},
\label{eq:q-from-nft}
\end{IEEEeqnarray}
where we used \eqref{eq:ab-norm} (alternatively, see \cite[Eq.~24]{yousefi2012nft2}).
Equation \eqref{eq:q-from-nft} constitutes the recovery relation. 

The derivative term $\partial_t b$ in \eqref{eq:q-from-nft} incurs numerical error. It is preferred to recover 
$q(t)$ via a relation that does not involve derivative. 
From \cite[Eq. 32]{yousefi2012nft1}, we have
\begin{IEEEeqnarray}{rCl}
q^*(t)&=&\frac{s}{\pi}\int\limits_{-\infty}^{\infty}\hat q(\lambda)e^{j2\lambda t}V_2^1(t,\lambda)\der \lambda,
\label{eq:q-from-ab-1}
\end{IEEEeqnarray} 
where $\vect{V}^1\eqdef [V_1^1, V_2^1]^T$ is a normalized
eigenvector with the boundary condition $\vect{V}^1(+\infty,\lambda)=(1,0)^T$  
\cite[Eqs. 27 \& 17(a)]{yousefi2012nft1}. 

Fix $t$ and define 
\begin{IEEEeqnarray}{rCl}
p_t(\tau)\eqdef
\begin{cases}
q(\tau), & \tau<t,\\[1pt]
\frac{1}{2}q(t), & \tau=t,\\[1pt]
0, & \tau> t.
\end{cases}
\label{eq:pt-def}
\end{IEEEeqnarray}
Applying \eqref{eq:q-from-ab-1} to $p_t(\tau)$ at $\tau=t$, we obtain 
\begin{IEEEeqnarray}{rCl}
p^*_t(t)&=&
\frac{s}{\pi}\int\limits_{-\infty}^{\infty}\hat p_t(\lambda)e^{j2\lambda t}\bar{V}_2^1(t,\lambda)\der \lambda
\nn
\\
&\overset{(a)}{=}&
\frac{s}{\pi}\int\limits_{-\infty}^{\infty}\hat p_t(\lambda)e^{j2\lambda t}\der \lambda
\label{eq:p-Fourier}
\\
&\overset{(b)}{=}&
\frac{s}{\pi}\int\limits_{-\infty}^{\infty}\hat q(t,\lambda)e^{j2\lambda t}\der \lambda,
\label{eq:pt}
\end{IEEEeqnarray}
where $\bar{\vect{V}}^1$ is an eigenvector of $p_t(\tau)$, 
$\hat p_t(\lambda)$ is the NFT of $p_t(\tau)$ and $\hat
q(t,\lambda)\eqdef b(t,\lambda)/a(t,\lambda)$.
Step $(a)$ follows because $p_t(\tau)=0$ for $\tau>t$, thus $\bar{V}^1_2(t,.)=\bar{V}^1_2(+\infty,.)=1$.
Step $(b)$ follows because $\hat q(t,\lambda)$ is in one-to-one relation with $q(\tau)=p_t(\tau)$ for 
$\tau<t$, thus $\hat p_t(\lambda)=\hat q(t,\lambda)$. From \eqref{eq:pt-def} and \eqref{eq:pt}
\begin{IEEEeqnarray*}{rCl}
q^*(t)=\frac{2s}{\pi}\int\limits_{-\infty}^{\infty}\hat q(t,\lambda)e^{j2\lambda t}\der \lambda.
\label{eq:qt}
\end{IEEEeqnarray*}

Function $p_t(\tau)$ is in general discontinuous at $\tau=t$.  
From \eqref{eq:p-Fourier}, $p^*_t(t)$ is described by a Fourier integral. 
At a point of discontinuity, the Fourier integral is the average of the value 
of the function on the two sides of that point. The choice $p(t)=q(t)/2$ in \eqref{eq:pt-def} 
ensures that \eqref{eq:p-Fourier} holds at the point of discontinuity. 

It follows that 
\begin{IEEEeqnarray}{rCl}
q^*[k]=\frac{2s}{\pi}\int\limits_{-\infty}^{\infty}\hat q[k,\lambda]e^{j2\lambda t[k]}\der\lambda,
\label{eq:CLP-q-from-qhat}
\end{IEEEeqnarray}
where $\hat q[k,\lambda]\eqdef\hat q(t(k],\lambda)$. The integral in \eqref{eq:CLP-q-from-qhat} can be discretized on the mesh
\begin{IEEEeqnarray}{rCl}
\lambda[m] \eqdef L_1+m\mu, \quad \quad m=0,\cdots,M-1,
\IEEEeqnarraynumspace
\label{eq:lambda-mesh}
\end{IEEEeqnarray}
where $\lambda\in [L_1,L_2]$, $L\eqdef L_2-L_1$, and $\mu\eqdef L/M$.

Steps 3) and 2) are consecutively performed to obtain $q[k]$, for all $k$. Step 1) is performed 
initially if $\vect u[k,z]$ is updated instead of $\hat q[k,\lambda]$.

\begin{figure*}[t]
\centerline{
\begin{tabular}{ccc}
\scalebox{1}{\begin{tikzpicture}[scale=0.666,font=\normalsize]

\def\fx{\x,{3*1/exp(((\x)^2)/4)}}

\draw[color=darkred,line width=1.1pt, samples=100, domain=-4:4] plot (\fx) node[right] {};

\foreach \y in {-3,-2.5,-2,-1.5,-1,-0.5,0,0.5}
{
\def\fy{3*1/exp(((\y)^2)/4)}
\draw[line width=1.1pt, color=black] ({\y},{0}) rectangle ({\y+0.5},{\fy});
}

\foreach \y in {1,1.5,2,2.5}
{
\def\fy{3*1/exp(((\y)^2)/4)}
\draw[dotted,line width=1.1pt, color=darkblue] ({\y},{0}) rectangle ({\y+0.5},{\fy});
}

\foreach \y in {-2.5,-2,-1.5,-1,-0.5, 0, 0.5,1,1.5,2,2.5,3}
{
\def\fy{3*1/exp(((\y-0.5)^2)/4)}
\filldraw[line width=1.1pt] ({\y},{\fy}) circle (0.05cm);
}

\filldraw[line width=1.1pt] (-3,0) circle (0.05cm);

\foreach \y in {-3,-2.5,-2,-1.5,-1,-0.5, 0, 0.5,1,1.5,2,2.5}
{
\def\fy{3*1/exp(((\y)^2)/4)}
\filldraw[line width=1.1pt, color=white] ({\y},{\fy}) circle (0.05cm);
}

\draw[line width=1.1pt, ->] (-4.5,0) -- (4.5,0) node[right] {$t$};

\def\deltaspace{-0.5}
\node at (1,\deltaspace){$t[k]$};
\node[right=-3mm] at (-1.5,4.2){$q^*[k]=\frac{2s}{\pi}\int\limits_{-\infty}^{\infty}\hat q[k,\lambda]e^{j2\lambda t[k]}\der\lambda$};

\draw[line width=1.1pt, <-] (1.28,2.82) -- (1.9,3.5);

\end{tikzpicture}}
&
\includegraphics[width=0.3\textwidth]{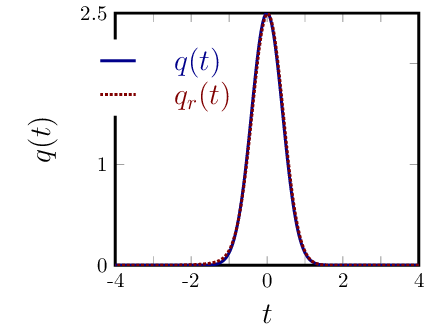}
&
\includegraphics[width=0.3\textwidth]{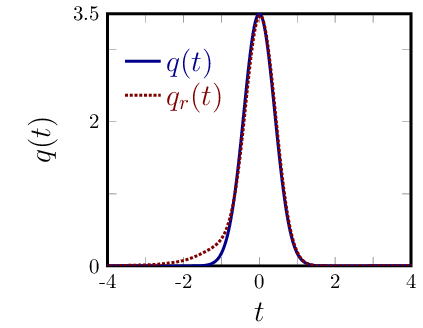}
\\
(a)
&
~~~~~~~~~~~(b)
&
~~~~~~~~~~~(c)
\end{tabular}
}
\caption{ The Boffetta-Osborne scheme: (a) schematic diagram; (b,c) 
the result $q_r(t)$ of applying the forward-inverse NFT to Gaussian functions $q(t)$ with two amplitudes. The parameters are $T=L=64$ and $N=M=2048$.}
\label{fig:piece-wise-constant}
\end{figure*}
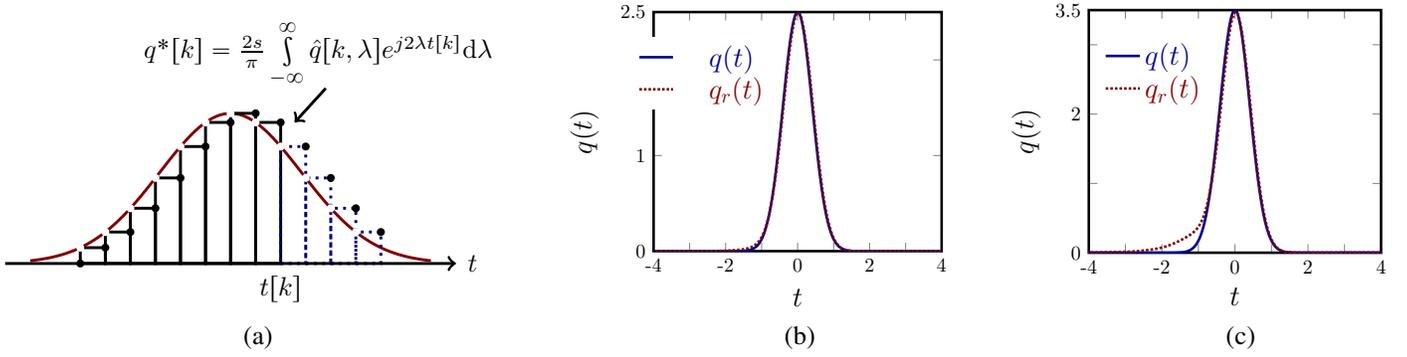

In Section~\ref{sec:remarks} we will show that the BO scheme, although suitable for the forward NFT, 
is inaccurate for the inverse NFT. As a result, below we consider the Ablowitz-Ladik scheme, where
all variables are discrete and finite-dimensional. 
We may also refer to the AL scheme as the discrete layer-peeling (DLP).

\subsection{Inverse Ablowitz-Ladik Scheme}
\label{sec:al}

The AL discretization of the eigenproblem \eqref{eq:eigen} is 
\cite[Eq. 17]{yousefi2012nft2}
\begin{IEEEeqnarray}{rCl}
\vect{v}[k+1,z]&=&c_k
\begin{pmatrix}
z^{\frac{1}{2}} & Q[k]\\
sQ^*[k] & z^{-\frac{1}{2}}
\end{pmatrix}
\vect{v}[k,z],
\label{eq:AL}\\
\vect{v}[0,z] &=&
\begin{pmatrix}
1\\
0
\end{pmatrix}
z^{\frac{k_0}{2}},\quad 0\leq k\leq N-1,
\nn
\end{IEEEeqnarray}
where $z\eqdef \exp(-2j\lambda\epsilon)$, $Q[k]\eqdef\epsilon q[k]$, $c_k
\eqdef 1/\sqrt{1-s|Q[k]|^2}$, and $k_0\eqdef T_1/\epsilon$. 
Equation \eqref{eq:ab-norm} suggests the following change of variable
\begin{IEEEeqnarray}{llllll}
a[k,z]\eqdef z^{-\frac{k_0+k}{2}}v_1[k,z],\quad b[k,z]\eqdef z^{\frac{k_0+k}{2}}v_2[k,z].
\label{eq:Ak-ak-v-a}
\IEEEeqnarraynumspace
\end{IEEEeqnarray}
The variable $z$ in this section should not be confused with the distance.

Upon iterating \eqref{eq:AL}--\eqref{eq:Ak-ak-v-a}, $a[k,z]$ and $b[k,z]$ are expressed as 
a linear combination of powers of $z^{\frac{1}{2}}$ and $z^{-\frac{1}{2}}$. We scale 
$a$ and $b$ to work with only the negative powers of $z$: 
\begin{IEEEeqnarray}{llllll}
A[k,z]\eqdef a[k,z], \quad  B[k,z]\eqdef z^{-(k_0+k)+\frac{1}{2}}b[k,z].
\IEEEeqnarraynumspace
\label{eq:Ak-ak-v-b}
\end{IEEEeqnarray}

From \eqref{eq:AL}, \eqref{eq:Ak-ak-v-a} and \eqref{eq:Ak-ak-v-b}, we obtain the forward iteration for the 
scaled AL scheme
\begin{IEEEeqnarray}{rCl}
\begin{pmatrix}
A[k+1,z]\\
B[k+1,z]
\end{pmatrix}
&=&
c_k
\begin{pmatrix}
 1 & Q[k]z^{-1} \\
sQ^*[k] & z^{-1} 
\end{pmatrix}
\begin{pmatrix}
A[k,z]\\
B[k,z]
\end{pmatrix},
\label{eq:AL-norm}
\IEEEeqnarraynumspace
\\
\begin{pmatrix}
A[0,z]\\
B[0,z]
\end{pmatrix}
&=&
\begin{pmatrix}
1\\
0
\end{pmatrix},\quad
0 \leq k\leq N-1.
\nn
\end{IEEEeqnarray}
Note that $A[k,z]$ and $B[k,z]$ are polynomials of $z^{-1}$ with \emph{finite degrees}. 

Instead of updating $A[k,z]$ and $B[k,z]$ in \eqref{eq:AL-norm} for each $z$,  
we can update the polynomial coefficients.
The $z$-transforms of $A$ and $B$ are
\begin{IEEEeqnarray}{rCl}
A[k,z]&=&\sum\limits_{m=0}^{\bar M-1}A_m[k]z^{-m},  
\IEEEyesnumber\IEEEyessubnumber\label{eq:Ak-Bk-series-a}\\
B[k,z]&=&\sum\limits_{m=0}^{\bar M-1}B_m[k]z^{-m},
\IEEEyessubnumber
\label{eq:Ak-Bk-series-b}
\end{IEEEeqnarray}
where $\bar M\in\naturals$ is the number of non-zero coefficients, 
and $A_m[k]$ and $B_m[k]$ 
are calculated as 
\begin{IEEEeqnarray}{rCl}
A_{m}[k]&=&\frac{1}{L}\int\limits_{L_1}^{L_2} A[k,e^{-2j\lambda t[k]}]e^{-j2m\epsilon\lambda}\der\lambda.
\label{eq:A_m[k]-}
\end{IEEEeqnarray}
If $a[k,z]=A[k,z]$ is obtained via \eqref{eq:AL-norm}, then $\bar M<\infty$, and $a[k,z]$ approximates $a(t,z)$. 
If $a[k,z]$ is obtained by  discretizing $a(t,\lambda)$ --- namely, $a[k,.]$ equals to 
$a(k\epsilon,.)$  for some $k$--- then in general $\bar M=\infty$.

The variable $z=z(\lambda)$ is the nonlinear frequency. The variable $m\in\bigl\{0,\cdots, \bar M-1\bigr\}$ 
may be viewed as the \emph{nonlinear time}, the 
Fourier-conjugate of the nonlinear frequency $z$. Define the vector of 
 the samples of $A[k,z(\lambda)]$ in the nonlinear frequency $\lambda$
\begin{IEEEeqnarray}{rCl}
\vect{\hat A}[k] \eqdef
\begin{pmatrix}
A[k,e^{-j2\lambda[0] t[k]}]
&
\cdots
&
A[k,e^{-j2\lambda[M-1] t[k]}]
\end{pmatrix},
\IEEEeqnarraynumspace
\label{eq:A[k]}
\end{IEEEeqnarray}
and the corresponding vector in the nonlinear time 
\begin{IEEEeqnarray*}{rCl}
\vect A[k] &\eqdef& \begin{pmatrix}A_0[k],\cdots, A_{\bar M-1}[k] \end{pmatrix}.
\end{IEEEeqnarray*}

Let us discretize the integral in \eqref{eq:A_m[k]-} in $\lambda$ on the mesh \eqref{eq:lambda-mesh}, 
choosing for the rest of the paper
\begin{IEEEeqnarray*}{rCl}
M=\bar M=N,\quad L=\pi/\epsilon,\quad \mu=\pi/T. 
\end{IEEEeqnarray*}
We obtain that $\vect{A}[k]$ and $\vect{\hat A}[k]$ are
the scaled nonlinear Fourier coefficients in the temporal and frequency domains, and 
\begin{IEEEeqnarray}{rCl}
\vect A=\frac{1}{M}
\vect{e}\hadamard
\textnormal{DFT}(\vect{\hat A}),
\label{eq:A_m[k]}
\end{IEEEeqnarray}
where DFT is the discrete Fourier transform (with zero-based indexing), $\hadamard$ is the Hadamard product, 
and
\begin{IEEEeqnarray*}{rCl}
\vect e\eqdef (1,e^{-j2\pi\frac{L_1}{L}}, \cdots, e^{-j2\pi\frac{L_1}{L}(M-1)}).
\end{IEEEeqnarray*}
Similar relations hold for the coefficient $B$.

The three steps of the inverse NFT via the AL scheme are as follows.

\subsubsection*{Step 1 (Factorization)}

First we compute  $a(N,\lambda) \eqdef a(\lambda)$ and $b(N,\lambda)\eqdef b(\lambda)$ 
from $Q(\lambda)$ according to the procedure 
outlined in Step 1 in the BO scheme. Then we calculate the frequency domain vectors
$\vect{\hat A}[N]$ and $\vect{\hat B}[N]$ from \eqref{eq:Ak-ak-v-b} and \eqref{eq:A[k]},
and the time domain vectors $\vect{A}[N]$ and $\vect{B}[N]$ 
from \eqref{eq:A_m[k]}. The coefficients $A_m[N]$ and $B_m[N]$ are generally non-zero for all $m\geq 0$. 
Here, these sequences must be truncated.

\subsubsection*{Step 2 (Integration)}

From \eqref{eq:AL-norm}, the backward iteration in the frequency domain $z$ is
\begin{IEEEeqnarray}{rCl}
\begin{pmatrix}
A[k,z]\\
B[k,z]
\end{pmatrix}
&=&
c_k
\begin{pmatrix}
 1 & -Q[k] \\
-sQ^*[k]z & z 
\end{pmatrix}
\begin{pmatrix}
A[k+1,z]\\
B[k+1,z]
\end{pmatrix},
\IEEEeqnarraynumspace
\label{eq:inverse-AB}
\end{IEEEeqnarray}
where $k=N-1,N-2,\cdots,0$. 
From \eqref{eq:inverse-AB}, the backward iteration in the temporal domain is
\begin{IEEEeqnarray*}{rCl}
\vect{A}[k] & =& c_k\Bigl(\vect A[k+1]-Q[k]\vect B[k+1]\Bigr),\\
\vect{B}[k] & =& c_k\shift\Bigl(-sQ^*[k]\vect A[k+1]+\vect B[k+1]\Bigr),
\end{IEEEeqnarray*}
where $\shift(\vect x)\eqdef (x_2,\cdots, x_n,0)$ is the left shift of $\vect x\eqdef(x_1,x_2,\cdots, x_n)$. 

The continuous-frequency unimodularity condition \eqref{eq:C-unimodularity} on 
the unit circle $|z|=1$ is
\begin{IEEEeqnarray}{rCl}
\bigl|A[k,z]\bigr|^2-s\bigl|B[k,z]\bigr|^2=1.
\label{eq:AA*-BB*=1-z}
\end{IEEEeqnarray}
In the temporal domain
\begin{IEEEeqnarray}{rCl}
  \vect A\convolution\cev{\vect{A}}^*-s \vect B\convolution \cev{\vect{B}}^* =\boldsymbol{\delta},
\label{eq:AA*-BB*=1}
\end{IEEEeqnarray}
where $\boldsymbol{\delta} \in\Reals^{2M-1}$ is the all-zero vector except for a one in the middle, 
$\convolution$ is vector convolution, 
and $~\cev{}~$ is the flip operator, \eg, $\cev{\vect{A}}_m=\vect A_{M-m-1}$. The condition \eqref{eq:AA*-BB*=1-z} or 
\eqref{eq:AA*-BB*=1} can be checked in iterations to monitor the numerical error.

\subsubsection*{Step 3 (Signal Recovery)}

The discretization of  $\hat q(t,\lambda)=\exp(-2j\lambda t)(v_2/v_1)$ on the time mesh \eqref{eq:t[k]} 
suggests to consider $\hat q[k,z] \eqdef z^{k_0+k} v_2[k,z]/v_1[k,z]$.
We have
\begin{IEEEeqnarray*}{rCl}
\hat q[k+1,z] &=&z^{k_0+k+1}\frac{v_2[k+1,z]}{v_1[k+1,z]}
\\
&=&z^{k_0+k+1}\frac{sQ^*[k]v_1[k,z]+z^{-\frac{1}{2}}v_2[k,z]}{z^{\frac{1}{2}}v_1[k,z]+Q[k]v_2[k,z]}
\\
&=&z^{k_0+k+\frac{1}{2}}\frac{sQ^*[k]+ z^{-\frac{1}{2}}\left(\dfrac{v_2[k,z]}{v_1[k,z]}\right)}{1+z^{-\frac{1}{2}}
Q[k]\left(\dfrac{v_2[k,z]}{v_1[k,z]}\right)}.
\end{IEEEeqnarray*}
By induction, we can see that $\deg(v_1)>\deg(v_2)$ for $k\geq 1$, 
where $\deg(v)$ is the highest positive power of $z$ in $v$. 
Thus 
$\lim\limits_{z\rightarrow \infty} v_2[k,z]/v_1[k,z]=0$ and 
\begin{IEEEeqnarray}{rCl}
Q^*[k]&=&s\lim\limits_{z\rightarrow\infty}z^{-k_0-k-\frac{1}{2}}\hat q[k+1,z]
\nn\\
&=&s\lim\limits_{z\rightarrow\infty}z^{-k_0-k-\frac{1}{2}}\frac{b[k+1,z]}{a[k+1,z]}
\nn\\
&=&s\frac{B[k+1,\infty]}{A[k+1,\infty]}
\nn
\\
&=&s\frac{B_0[k+1]}{A_0[k+1]}.
\IEEEeqnarraynumspace
\label{eq:qk-from-vk+1-ALL}
\end{IEEEeqnarray}
Equating \eqref{eq:qk-from-vk+1-ALL} basically follows from equating the like powers of 
$z$ in \eqref{eq:AL-norm}. 

The AL scheme is summarized in Algorithm~\ref{alg1}.

\begin{algorithm}[t]
\caption{The AL scheme for the inverse NFT.}
\label{alg1}
\begin{algorithmic}

\STATE Obtain $a(\lambda)$ and $b(\lambda)$ from $\hat q(\lambda)$, 
using \eqref{eq:|a|} and \eqref{eq:angle(a)}.
Set $a[N,\lambda]\eqdef a(\lambda)$ and
$b[N,\lambda] \eqdef b(\lambda)$

\STATE Compute $\vect{\hat A}[N]$ and $\vect{\hat B}[N]$ from \eqref{eq:Ak-ak-v-b} and \eqref{eq:A[k]},
and $\vect{A}[N]$ and $\vect{B}[N]$ 
from \eqref{eq:A_m[k]}. 
Truncate $\vect{A}[N]$ and $\vect{B}[N]$ to finite-dimensional vectors.

\FOR{ $k=N, N-1,\cdots, 1$}

\STATE Obtain $Q[k]$ from \eqref{eq:qk-from-vk+1-ALL} or \eqref{eq:Q[k-1]}, and $q[k]=Q[k]/\epsilon$.

\STATE Update
\begin{IEEEeqnarray*}{rCl}
\vect A &\leftarrow& c_k\bigl(\vect A-Q[k]\vect B\bigr),\\
\vect B &\leftarrow& c_k\shift\bigl(-sQ^*[k]\vect A+\vect B\bigr).
\end{IEEEeqnarray*}

\ENDFOR
\end{algorithmic}
\end{algorithm}

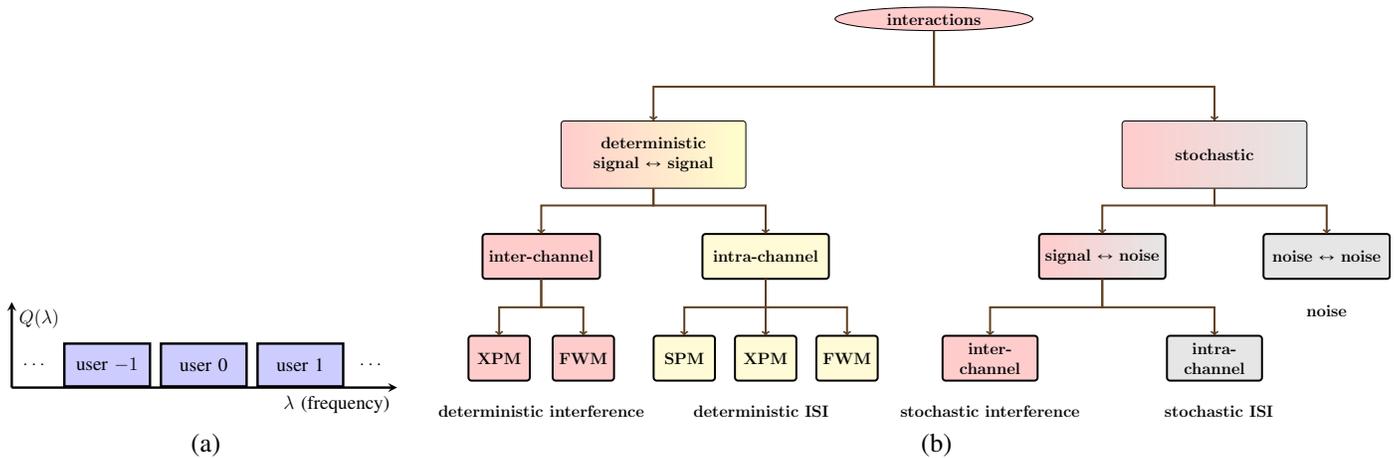
\begin{figure*}[t]
\centerline{
\begin{tabular}{cc}
\scalebox{0.6}{\begin{tikzpicture}[>=stealth,
ubox/.style={draw,line width=1.5pt, fill=blue!20,minimum height=5ex,minimum width=10ex}]

\node[ubox] (O) {user 0};
\node[ubox,right=1ex of O] (R) {user 1};
\node[ubox,left=1ex of O] (L) {user $-1$};
\node[left=1ex of L] (CL) {$\cdots$};
\node[right=1ex of R] (CR) {$\cdots$};

\draw[->,line width=1.5pt] (L.south west) ++(-6ex,0) coordinate (ORIG) -- ++(45ex,0) node[anchor=east,yshift=-2ex] {$\lambda$ (frequency)};
\draw[->,line width=1.5pt] (ORIG) -- ++(0,10ex) node[anchor=north west]{$Q(\lambda)$} node[anchor=east,inner sep=1pt]{};

\end{tikzpicture}}
&
\includegraphics[width=0.7\textwidth]{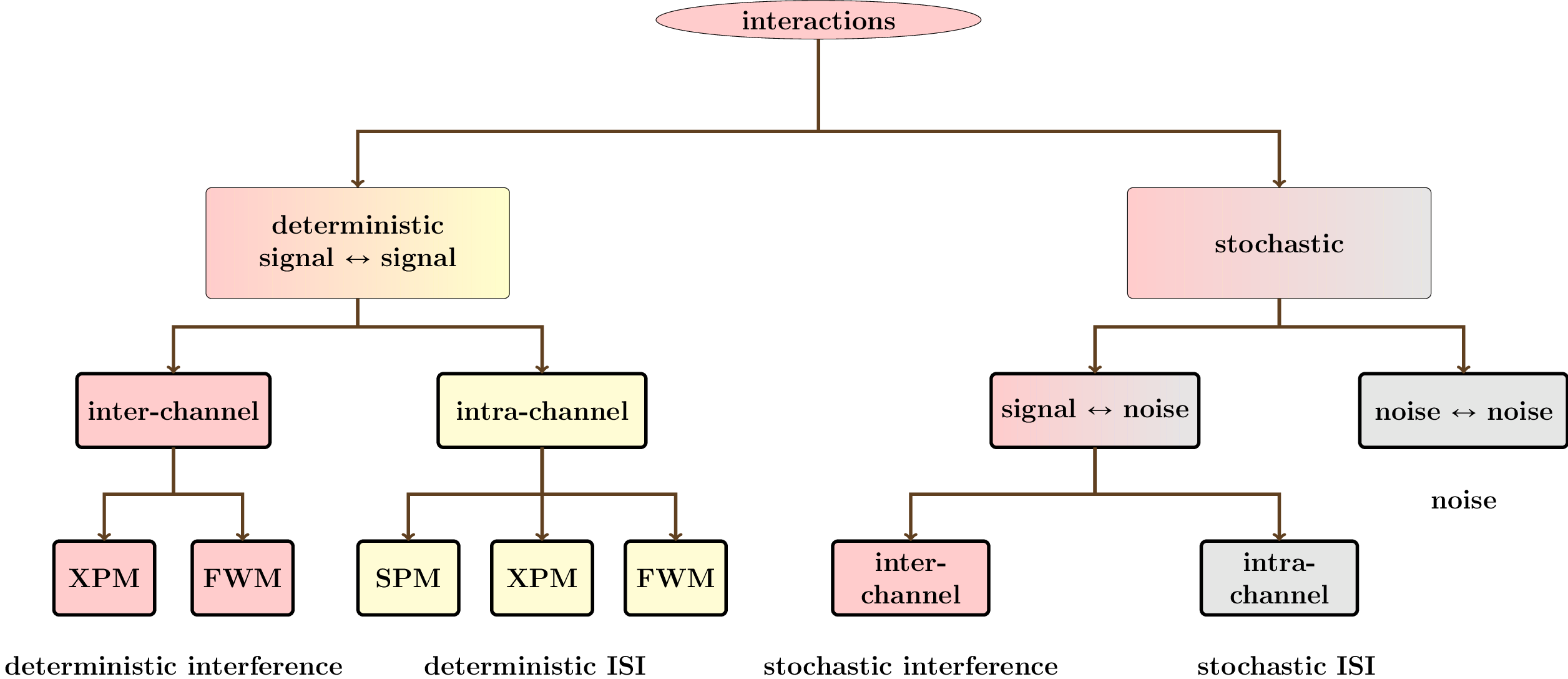}
\\
(a) & ~~~~~(b)
\end{tabular}
}
\caption{ (a) Multi-user NFDM. (b) Interactions in linear
  multiplexing. Interactions that cannot be removed are identified
  with red. Interactions that can partly be equalized using signal
  processing are shown by yellow. Grey terms are addressed by coding. SPM and XPM stand for self- and cross-phase
modulation, and FWM is four-wave mixing.}
\end{figure*}



\subsection{Complexity, Accuracy and Remarks}
\label{sec:remarks}

\subsubsection{Complexity}

Computing the Fourier integral or the NFT amounts to performing an integration; see \eqref{eq:eigen}.
There are two variables $t$ and $\lambda$. If each variable is discretized in a mesh with $N$ points, there are $N^2$
points in a rectangular mesh in the $t-\lambda$ plane. Therefore, the basic algorithms for computing the 
Fourier transform or the NFT take $O(N^2)$ operations --- including the AL and BO schemes in this paper.
This complexity, of course, is reduced to $O(N\log N)$ in the 
Fourier transform and, in some cases, to $O(N\log^2 N)$ in the NFT \cite{skaar2003,vaibhav2017finft}. 
The least complexity NFT is subject to ongoing research.

\subsubsection{Accuracy}

There are two sources of error in the BO scheme. First, the function $q(t)$ is approximated by 
a piece-wise constant function --- once this approximation is made, the BO scheme is exact 
for the piece-wise constant function. The error in this part is $O(\epsilon)$ for smooth $q(t)$. 
Second, the integral \eqref{eq:CLP-q-from-qhat} has to be discretized which is 
subject to error. Furthermore, here the discontinuity of $q(t)$ results in numerical errors when evaluating the Fourier integral 
\eqref{eq:CLP-q-from-qhat} due to the Gibbs effect.  Consequently, Step 3 is inexact. The error in this part depends on how the integral 
\eqref{eq:CLP-q-from-qhat} is evaluated. Steps 1 and 2, on the other hand, are exact in the BO scheme.

Figs.~\ref{fig:piece-wise-constant} (b--c) show the result of applying the forward NFT followed by the 
inverse NFT  to a Gaussian function using the BO scheme in the defocusing regime. It can be seen that the error is increased with 
the signal amplitude. Importantly, the local error accumulates as the iteration proceeds towards $t=T_1$.
The BO scheme is sensitive to i) spectral parameters $L$ and $M$ and, to a smaller extent, 
temporal variables $T$ and $N$; 
and ii) the accuracy of the computing the integral \eqref{eq:CLP-q-from-qhat}.

In the AL scheme, Steps 2 and 3 are exact. 
Here too there are two sources of error. 
First, the continuous-time operator $L$ is approximated by its AL discretization. The error in this
part is $O(\epsilon)$ for smooth functions. Second, the vectors $\vect A[N]$
and $\vect B[N]$ are truncated to finite-dimensional vectors in Step 1. The error in this part depends on the 
coefficients that are neglected.
Note that Steps 2 and 3 are performed consecutively. Since these steps 
are exact, 
the error does not accumulate in the AL scheme. The step that is subject to error is Step 1 which is 
outside the iteration. Therefore, the error can be controlled in the AL scheme.

\subsubsection{Remarks}

The AL discretization is obtained by approximation $1-2j\lambda$ in the first-order Euler discretization 
with $\exp(-j2\lambda \epsilon)$ \cite[Sec.~III. E]{yousefi2012nft2}.
This approximation makes eigenvectors a periodic 
function of $\lambda$ with period $L=\pi/\epsilon$. This is in
contrast to the BO scheme where $L$ is arbitrary.

In the forward AL scheme, $\vect A[k]$ and $\vect B[k]$ have finite lengths and 
satisfy  the discrete-frequency unimodularity condition \eqref{eq:AA*-BB*=1-z}.
However, in the inverse NFT, while $a(\lambda)$ and $b(\lambda)$ are obtained from the
continuous-frequency unimodularity condition \eqref{eq:C-unimodularity}, 
the truncated $\vect A[N]$ and $\vect B[N]$ do not generally satisfy the discrete-frequency unimodularity 
condition \eqref{eq:AA*-BB*=1}.
In general, $\vect A[N]$ and $\vect B[N]$ should be realizable, \ie, they should be the image of 
$\vect A[0]$ and $\vect B[0]$ for some signal \cite{skaar2003}.  
The AL scheme is improved if  \eqref{eq:AA*-BB*=1-z} is enforced in the backward iteration. 
It has been shown that this minor modification notably improves the algorithm, because it prevents the numerical 
error from propagating in the iteration \cite{skaar2003}.

In the AL scheme, $\vect v[N,z]$ depends on $\vect v[0,z]$ with the scale factor $\prod c_k=\exp(P)$, where
\begin{IEEEeqnarray*}{rCl}
  P\eqdef -\frac{1}{2}\sum\limits_{k=1}^N\log\left(1-s\left|Q[k]\right|^2\right).
\end{IEEEeqnarray*}
As $N\rightarrow \infty$,  $\epsilon$, $Q[k]$ and $P$ tend to zero. 
As a heuristic condition, we require that
\begin{IEEEeqnarray}{rCl}
0\leq P<P_{\max} 
\label{eq:P<Pmax}
\end{IEEEeqnarray}
for a moderate $P_{\max}$, and
\begin{IEEEeqnarray}{rCl}
  |Q[k]|<1.
\label{eq:|Q[k]|<1}
\end{IEEEeqnarray}
It can be shown that
\begin{IEEEeqnarray*}{rCl}
  \frac{\der|Q[k-1]|}{\der |Q[k]|}\propto \frac{1}{\left(1-s|Q[k]|^2\right)^2}.
\end{IEEEeqnarray*}
It is thus desirable to have 
\begin{IEEEeqnarray}{rCl}
 |Q[k]|\ll 1, 
\label{eq:|Q[k]|<0.1}
\end{IEEEeqnarray}
so that $Q[k-1]$ is not 
sensitive to error in $Q[k]$.

The AL scheme  can  be implemented in the frequency domain $z$ as well.  Instead of updating
$\vect{A}[k]$ and $\vect{B}[k]$, one could update $\vect{\hat A}[k]$ and $\vect{\hat B}[k]$. 
The accuracy of the AL scheme stems from the fact that it works with finite-dimensional vectors, and 
from the exactness of the Step 2 and 3, not from the implementation in the temporal domain.

The recovery relation \eqref{eq:qk-from-vk+1-ALL} can be obtained in alternative ways.
In Appendix~\ref{sec:A-B}, both \eqref{eq:qk-from-vk+1-ALL} and 
the following recovery relation are proved by induction
\begin{IEEEeqnarray}{rCl}
 Q[k]=\frac{A_{k}[k+1]}{B_{k}[k+1]}.
\label{eq:Q[k-1]}
\end{IEEEeqnarray}
Finally,  \eqref{eq:qk-from-vk+1-ALL} can also be obtained via discretizing \eqref{eq:q-from-nft},
by replacing $\exp(j2\lambda
t)$ with $z^{-k-k_0}$, $\partial_t$ with $z^{1/2}/\epsilon$, and
equating the zeroth-order term in the $z$ expansion of both sides (corresponding to
$z\rightarrow\infty$ or $\lambda\rightarrow\infty$ in $\Complex^+$).
Both \eqref{eq:qk-from-vk+1-ALL} and \eqref{eq:Q[k-1]} follow basically by equating the like powers of 
$z$ in $v[k,z]$ in the AL scheme.

The BO scheme that is proposed in \cite{boffetta1992cds} is a first-order Euler integration 
scheme --- except that, for the Zakharov-Shabat system in question, each
step in integration can be done analytically without approximation. 
It generally applies well to differential equations which are exactly solvable in the special case 
that coefficients are constant. 

The BO and AL schemes are simple first-order integration schemes. They can be extended to higher-order 
integration schemes in straightforward ways 
to compute the NFT arbitrary accurately (at the expense of performing more operations).  
We do not pursue more accurate or faster algorithms in this paper.


\section{Linear and Nonlinear Frequency-Division Multiplexing}
\label{sec:nfdm}

The theoretical underpinning of NFDM with discrete and continuous spectrum  is presented in 
\cite{yousefi2012nft1,yousefi2012nft2,yousefi2012nft3}. However, the theory simplifies considerably when the discrete spectrum 
is absent. 
In this section, we simplify NFDM with the continuous spectrum and present a general approach to nonlinear
modulation and multiplexing (valid in the focusing and defocusing regimes).

\subsection{Linear Modulation and Multiplexing}

Let $\mathcal H$ be a complex Hilbert
space with a basis $(\phi_{\ell})_{\ell}$, $\norm{\phi_{\ell}}=1$. 
Linear modulation on 
$\mathcal H$ corresponds to $x\eqdef\sum_{\ell} s_{\ell}\phi_{\ell}$, where $s_{\ell}\in\Complex$ 
are symbols. If the basis is orthogonal, the demodulation is performed efficiently 
per-symbol as $s_{\ell}=\inner{x}{\phi_{\ell}}$. 
The basis elements are sometimes called (linear-algebraic) transmission modes. 

Let $\mathcal H_k$ be disjoint subspaces of $\mathcal H$ (excluding the zero element). 
Linear multiplexing of $x_k\in\mathcal H_k$ corresponds 
to $x\eqdef\sum_{k} x_k$. If $\mathcal H_k$ are pair-wise orthogonal, the demultiplexing is 
simply $x_k=P_{\mathcal H_k}(x)$, where $P_{\mathcal H_k}$ is the orthogonal projection 
onto $\mathcal H_k$.  Clearly, $x_k$ need not be linearly modulated within $\mathcal H_k$.

Signals that are orthogonal at the channel input may not be orthogonal at the channel output. 
As a result, linear modulation and multiplexing are generally subject to ISI and interference
--- even in the linear channels.  Suppose that the channel is governed by a compact (linear)
self-adjoint map $T$ in the absence of noise. Then 
$(\phi_{\ell})_{\ell}$ can be chosen to be the set of 
orthogonal eigenvectors of $T$, for which the ISI is absent under per-symbol demodulation. 
Similarly, if $\mathcal H_k$ are chosen to be the span of disjoint subsets of orthogonal eigenvectors of $T$, interference 
is absent in demultiplexing by projection.

In linear modulation and multiplexing the set of signals that are transmitted in the 
channel (\ie, in the time domain) is represented by a vector space. We next consider
NFDM where this signal space is not a linear space.  

\subsection{Nonlinear Modulation and Multiplexing Using the Continuous Spectrum}

Nonlinear modulation and multiplexing using the NFT consists of linear modulation and multiplexing
in the nonlinear Fourier domain. In what follows, we consider a multi-user system with $N_u$ users 
and $N_s$ symbols per user. Linear and nonlinear (passband) bandwidths are denoted, respectively, 
by  $B$ and $W$.

\subsubsection{Signal Spaces}

Let $\hat{\mathcal{Q}}$ be the space of signals $\hat
q(\lambda,z): \mathcal W \times \Realsnn \mapsto \mathbb D$, where
$\mathcal W\eqdef [-\pi W,\pi
W]$, and
\begin{IEEEeqnarray*}{rCl}
  \mathbb D\eqdef
  \begin{cases}
    \Complex, & s=-1,
    \\
    \mathbb T, & s=+1,
    \end{cases}
\end{IEEEeqnarray*}
where $\mathbb T$ is the open unit disk $\mathbb T\eqdef\bigl\{
z\in\Complex \: : \: |z|<1 \bigr\}$. 
Note that, from \eqref{eq:parseval}, if $s=1$, $|\hat{q}(\lambda,z)|<1$.

In the defocusing regime, $\hat{\mathcal{Q}}$ is not a vector
space (with addition and multiplication), and not best suited for a communication theory.
We introduce a transformation $F: \hat q(\lambda,z)\mapsto U(\lambda,z)$ so as to map $\mathbb D$ 
to $\Complex$. 
An appropriate choice motivated by \eqref{eq:parseval} is 
\begin{IEEEeqnarray}{rCl}
U(\lambda,z)\eqdef \left(-2s\log(1-s\left|\hat
  q(\lambda,z)\right|^2)\right)^{\frac{1}{2}}e^{j\angle(\hat q(\lambda,z))},
  \IEEEeqnarraynumspace
\label{eq:X-qhat}
\end{IEEEeqnarray}
where $\angle(\hat q)$ is the phase of $\hat q$.
We choose the image $\mathcal U$ of $F: \hat{\mathcal Q}\mapsto \mathcal U$
to be the vector space of finite-energy signals supported 
on $\mathcal W$. Transformation \eqref{eq:X-qhat} is required in the defocusing regime. In the 
focusing regime, it is not required, however it makes the signal energy in time, frequency, nonlinear time and
nonlinear frequency the same according to \eqref{eq:energies}. 
This property is convenient for modulation.

\subsubsection{NFDM Transmitter}
\label{sec:nfdm-transmitter}

Partition the space $\mathcal U=\bigcup_k
\mathcal U_k$ into orthogonal user subspaces $(\mathcal U_k)_{k=k_1}^{k_2}$, where 
user $k$ operates in the subspace $\mathcal U_k$ spanned by the orthogonal basis 
$(\Phi_{\ell}^k(\lambda))_{\ell=\ell_1}^{\ell_2}$, $\norm{\Phi_{\ell}^k(\lambda)}^2=2\pi$. 
Here, $k$ is the user index, $\ell$ is the symbol index and
\begin{IEEEeqnarray*}{rCl}
k_1\eqdef-\Bigl\lfloor\frac{N_u}{2}\Bigr\rfloor,\quad
k_2\eqdef\Bigl\lceil\frac{N_u}{2}\Bigr\rceil-1,
\end{IEEEeqnarray*}
where $\lfloor x \rfloor$ and
$\lceil x \rceil$ denote, respectively, rounding $x\in\Reals$ to
nearest integers towards minus and plus infinity. Similar relations hold for 
$\ell_1$ and $\ell_2$ in terms of $N_s$.

The transmitted signal in the nonlinear frequency domain is
\begin{IEEEeqnarray}{rCl}
  U(\lambda,0)\eqdef
  \underbrace{\sum\limits_{k=k_1}^{k_2}}_{\textnormal{linear mux}}
\underbrace{\left(\sum\limits_{\ell=\ell_1}^{\ell_2}s_\ell^k\Phi_\ell^k(\lambda)\right)}_{\textnormal{linear
mod.}},
\label{eq:wdm-signal}
\end{IEEEeqnarray}
where $(s_{\ell}^k)_{\ell=\ell_1}^{\ell_2}$ are symbols of the user $k$. As discussed 
in Section~\ref{sec:nft-optical}, the modulation in \eqref{eq:wdm-signal} can be linear or nonlinear, whereas the 
orthogonal multiplexing is required in NFDM.

In general, user signals can overlap in any domain. Interference
in user $k$ is in the orthogonal subspace $\mathcal U_k^\perp$ and can
be projected out. However, for the rest of the paper,  we consider  the special case where 
users operate in equally-spaced non-overlapping intervals of width
$W_0\eqdef W/N_u$ Hz in the nonlinear frequency 
$l\eqdef\lambda/2\pi$. User $k$ is centered
at the nonlinear frequency $l=kW_0$ and operates in the nonlinear frequency interval 
\begin{IEEEeqnarray*}{rCl}
  \mathcal W_k\eqdef\bigl[kW_0-\frac{W_0}{2} \:,\: k W_0+\frac{W_0}{2}\bigr],\quad k_1\leq k\leq k_2.
\end{IEEEeqnarray*}
In this case, $\Phi^k_\ell(\lambda)\eqdef\Phi_\ell(\lambda-2\pi k W_0)$ where 
$(\Phi_{\ell} (\lambda))_{\ell=\ell_1}^{\ell_2}$ is an orthogonal basis for signals supported on 
$\mathcal W_0$. 

Taking the inverse Fourier transform $\mathcal F^{-1}$ of \eqref{eq:wdm-signal} with respect to $\lambda$
\begin{IEEEeqnarray}{rCl}
u(\tau,0)&\eqdef&\ft^{-1}(U(\lambda,0)) \label{eq:u-U1}
\\
 &=&\sum\limits_{k=k_1}^{k_2}
\left(\sum\limits_{\ell=\ell_1}^{\ell_2}s_\ell^k\phi_\ell(\tau)\right)e^{j2\pi k W_0 \tau},
\label{eq:u-U}
\end{IEEEeqnarray}
where $\phi_{\ell}(\tau) \eqdef \mathcal F^{-1}\bigl(\Phi_{\ell}(\lambda)\bigr)$, and $\norm{\phi_{\ell}^k(\tau)}=1$. 
The variable $\tau$ can be interpreted
as the \emph{nonlinear time} (measured in seconds), the Fourier-conjugate of the nonlinear 
frequency $l=\lambda/2\pi$ (measured in Hz). It coincides with the physical time $2t$ for
small amplitude signals $\norm{q(t)}_{L^1}\ll 1$, $\tau\approx 2t$. Typically, 
\begin{IEEEeqnarray}{rCl}
\phi_\ell(\tau)\eqdef\phi(\tau-\ell T_0), 
\label{eq:phi-tau}
\end{IEEEeqnarray}
where the pulse shape $\phi(\tau)$ and $T_0$ are chosen so that $|\mathcal{F}(\phi)(f)|^2$  satisfies the Nyquist 
zero-ISI criterion for $T_0$.

The modulation begins with choosing the symbols $s_\ell^k$ from a constellation $\Xi$ and computing $U(\lambda,0)$. 
This can be done directly in the spectral domain
\eqref{eq:wdm-signal}, or starting in the temporal domain \eqref{eq:u-U}. The transmitted NFT signal is
\begin{IEEEeqnarray}{rCl}
\hat q(\lambda,0)=\left(s-se^{-\frac{s}{2}|U(\lambda,0)|^2}\right)^{\frac{1}{2}}e^{j\angle(U(\lambda,0))},
\label{eq:qhat-X}
\end{IEEEeqnarray}
where $U(\lambda,0)=\ft(u(\tau,0))$. Finally, 
\begin{IEEEeqnarray*}{rCl}
  q(t,0)=\inft(\hat q(\lambda, 0)).
\end{IEEEeqnarray*}
The symbols $s_{\ell}^k$ and the nonlinear bandwidth $W$ are chosen such that the Fourier spectrum 
$\mathcal Q(f,0)\eqdef\ft(q(t,0))$ has bandwidth $B$
Hz. When $\norm{q(t,0)}_{L^1}\ll 1$, $q(t,0)\approx -\sqrt 2 u^*(2t,0)$ and $W\approx
2B$.

The mapping \eqref{eq:X-qhat} ensures that $[s_\ell^k]$ and $q(t,0)$ have the same energies:
\begin{IEEEeqnarray}{rCl}
\sum\limits_{\ell,k}^{}|s_\ell^k|^2&=&\norm{u(\tau,0)}^2_{L^2(\Reals)}
\nn
\\
&=&\frac{1}{2\pi}\norm{U(\lambda,0)}^2_{L^2(\Reals)}
\nn
\\
&=&
-\frac{s}{\pi}\int\limits_{-\infty}^{\infty}\log\left(1-s\left|\hat
q(\lambda,0)\right|^2\right)\der\lambda
\nn
\\
&=&\int\limits_{-\infty}^{\infty}|q(t,0)|^2\der t,
\label{eq:energies}
\end{IEEEeqnarray}
where \eqref{eq:energies} follows from the Parseval's identity \eqref{eq:parseval}.

\subsubsection{NFDM Receiver}

At the receiver, first the NFT is applied to $q(t, \const{L})$ to obtain $\hat q(\lambda,\const L)$. Then, channel equalization is
performed
\begin{IEEEeqnarray}{rCl}
\hat q_e(\lambda,\const L)\eqdef H^{-1}(\lambda,\const L)\hat q(\lambda,\const L),
\label{eq:equal}
\end{IEEEeqnarray}
where $H(\lambda,\const L)$ is the channel filter in \eqref{eq:channel-filter}.
Next, $U_e(\lambda,\const L)$ is computed from $\hat q_e(\lambda,\const L)$ according to 
\eqref{eq:X-qhat}. Finally, the received symbols are
\begin{IEEEeqnarray}{rCl}
  \hat s_\ell^k=\frac{1}{2\pi}\int\limits_{-\infty}^{\infty} U_e(\lambda,\const
  L)\Phi_\ell^*(\lambda-2\pi kW_0)\der\lambda.
\label{eq:s-hat}
\end{IEEEeqnarray}

Alternatively,  $u_e(\tau,\const L)\eqdef\mathcal{F}^{-1}(U_e(\lambda,\const{L}))$ can be computed.
The received symbols at the output are obtained by match filtering 
\begin{IEEEeqnarray*}{rCl}
  \hat s_\ell^k=\int\limits_{-\infty}^{\infty} u_e(\tau,\const L)\phi_\ell^*(\tau)e^{-j2\pi k W_0\tau}\der\tau.
\end{IEEEeqnarray*}

\subsubsection{Uniform and Exponential Constellations}

We choose a  multi-ring constellation $\Xi$ for symbols
$s_\ell^k$ in the $U$ domain, as shown in Fig.~\ref{fig:clouds}(a). 
There are $N_r$ uniformly-spaced rings
with radii in interval $[a, b]$ and ring spacing $\Delta r$, and $N_{\phi}$ uniformly-spaced phases.
Under the transformation \eqref{eq:qhat-X}, $\Xi$ maps to an exponential constellation $\Xi'$ in the $\hat q$
domain with phases in $\Xi$ and radii
\begin{IEEEeqnarray}{rCl}
r_n^2\eqdef 1-e^{-\frac{1}{2} (\Delta r)^2 n^2},\quad 1\leq n\leq N_r,
\label{eq:rn-lambda}
\end{IEEEeqnarray}
where we assumed $s=1$ and $a=0$. Thus, in the defocusing regime, as $n\rightarrow\infty$, 
$r_n\rightarrow 1$ and the distance between rings in $\Xi'$ decreases
exponentially. In this way, an infinite number of
choices is realized in the finite interval $|\hat q| \in [0,1)$.

\begin{remark}
The form of the NFDM 
signal \eqref{eq:u-U} with $\tau$ is identical to the form of the WDM signal in \cite[Eq. 2]{yousefi2012nft3} 
with $t$. This expression is simply the sampling theorem, representing a signal with 
a finite linear or nonlinear bandwidth in terms of its DoFs. The DoFs are identified 
in the $\tau$ domain in \eqref{eq:u-U}, and equivalently in the 
$\lambda$ domain in \eqref{eq:wdm-signal}.

\qed
\end{remark}


\section{Comparison of the AIRs of WDM and NFDM}
\label{sec:simulations}

In this section, we compare WDM and NFDM in the defocusing regime, subject to the 
same bandwidth and average power constraints. 

The average power of the signal is defined as $\const P\eqdef E/T$, where $E$ and $T$ are the energy and time duration 
of the (entire multiplexed) signal $q(t,0)$, respectively. In this section, time duration and
bandwidth are defined as intervals containing $99\%$ of the signal
energy  \cite{yousefi2012nft3}. 
We choose $\phi(\tau)$ in \eqref{eq:phi-tau} a root raised-cosine 
function with the excess bandwidth factor denoted by $r$. 
The system parameters are given in Table~\ref{tab:params}. 
Importantly, to keep the computational complexity manageable, we choose 
one symbol per user, \ie, $N_s=1$.

We consider one ADM at the receiver. This is a filter in the frequency in WDM, and in the nonlinear
frequency in NFDM (the ADM  only drops signals). In WDM, back-propagation is applied to the filtered 
signal according to the NLS equation. The corresponding equalization in NFDM is given by \eqref{eq:equal}.

We first present a simple simulation to 
illustrate the main ideas, before comparing the AIRs and the SEs.

\subsubsection{Illustrative Example}
\label{sec:example}

Figs.~\ref{fig:demonstrations}(a)--(b) show a sample WDM and NFDM signals at the transmitter and receiver, in the
absence of noise.
These two signals have the same power, bandwidth and time duration at the transmitter;
see Figs. \ref{fig:NFDM-WDM-comparisons}(a)--(b). 
It can be seen that
WDM users' signals interfere with one another, while NFDM users' signals are perfectly separated. 
Fig.~\ref{fig:demonstrations}(c) shows the input output signals in WDM after equalization.
The distortion in Fig.~\ref{fig:demonstrations}(c) increases with $\const P$, $N_s$ and the 
number of ADMs. This distortion is the bottleneck
in linear multiplexing, because it cannot be mitigated in a network environment. 
The absence of this distortion in NFDM is the underlying reason that NFDM 
outperforms WDM. We recover the matrix of symbols $[s_{\ell}^k]$ in NFDM nearly
perfectly when noise is zero, for $N_u=N_s=15$.

\subsubsection{Achievable Information Rate}
\label{sec:wdm-nfdm-AIRs}

We approximate the channels after equalization by 
discrete memoryless channels $s_0^0\mapsto \hat s_0^0$ in the linear or nonlinear frequency. The AIR 
is defined as the maximum of the mutual information over the probability distribution $p(s_0)\eqdef p_{S_0}(s_0)$: 
\begin{IEEEeqnarray*}{rCl}
&&R(\const P) \eqdef \max\limits_{p(s_0)} ~I(s_0; \hat s_0), \\
 && \E |s_0|^2=\const P,
\end{IEEEeqnarray*}
measured in bits per two real (one complex) dimensions (bits/2D) \cite[Chap.~2]{forney2005course}.

The constellation $\Xi$ consists of $N_r =32$ or $64$ rings (depending on the power) each with $N_\phi=128$
phase points, where for NFDM $a=0$ and $b=1.6$. 
The power spectral density of the distributed noise is
$\sigma_0^2\eqdef n_{sp}h f_0\alpha=6.48\times 10^{-21}~\rm{W/(km.Hz)}$  
calculated with realistic parameters in Table~\ref{tab:params}. The bandwidth of the noise is set to be 
the maximum bandwidth of the signal in distance, which we assume is  
the bandwidth of a signal with the highest energy corresponding to $|s_\ell^k|=b$, $\forall k,\ell$. 
The number of signal samples in time and nonlinear
frequency is $N=M=16384$. 
We estimate the transition probabilities $s_0^0\mapsto
\hat s_0^0$ based on 4000 simulations of the stochastic NLS
equation.

Fig. \ref{fig:nfdm-wdm-rates} shows the AIRs of NFDM and WDM. As expected, 
the WDM AIR characteristically vanishes as the input power is increased more
than an optimal value $\const P^*\approx -10$ dBm. In contrast, the NFDM AIR continues to 
increase for $\const P>\const P^*$ --- at least up to the maximum power in Fig. \ref{fig:nfdm-wdm-rates} 
where we could perform simulations. The channel capacity is upper bounded by 
 $\log_2(1+\snr)$, where $\snr=\const P/(\sigma^2_0B\const L)$ is the 
signal-to-noise ratio. This upper bound in Fig.~\ref{fig:nfdm-wdm-rates} is not a 
perfect straight line, because the power in the horizontal axis is 
based on the $99\%$ time duration.

\begin{table}[t]
\caption{Fiber and System Parameters}
\label{tab:params}
\centerline{\begin{tabular}{c|l|l}
$n_{\rm sp}$ & 1.1 & {\footnotesize excess spontaneous emission factor}\\
$h$ & $6.626 \times 10^{-34} {\rm J} \cdot {\rm s}$ & {\footnotesize Planck's constant} \\
$f_0$ & 193.55~{\rm THz} & {\footnotesize center frequency} \\
$\alpha$ & $0.046~{\rm km}^{-1}$ & {\footnotesize fiber loss (0.2~dB/km)} \\
$\gamma$ & $1.27~{\rm W}^{-1}{\rm km}^{-1}$ & {\footnotesize nonlinearity parameter}\\ 
$\mathcal L$ & 2000 km & fiber length \\
$D$ & -17 ${\rm{ps/(nm.km)}}$ & {\footnotesize dispersion parameter}\\
$N_u$ & 15 & number of users \\
$N_s$ & 1 & number of symbols per user \\
$B$ & 60 GHz & total bandwidth \\
$r$ & 0.25\% & excess bandwidth factor
\end{tabular}}
\end{table}

\begin{figure}
\centering
\def\powernoise{31.08}
\pgfplotstableread{
SNR AIR
30.72 0
28.77 0   
26.27 0   
22.75 0   
16.82 0   
11.29 0    
5.45 0
}\snraxis

\pgfplotstableread{
P AIR
-0.28 13.9
-2.23 13.2
-4.73 12.39
-8.25 11.22
-14.18 9.57
-19.71 7.88
-25.55 5.96
}\awgnrate

\pgfplotstableread{
P AIR
-2.45 11.29
-3.56 11
-6.01 10.54
-9.6 9.73
-15.67 8.06
-19.75 7.07
-25.6 5.8
}\nfdmrate

\pgfplotstableread{
P AIR
-0.28 5.38 
-2.23 6.36
-4.73 7.13
-8.25 8.27
-14.18 8.24
-19.71 6.88
-25.55 5.5
}\wdmrate

\begin{tikzpicture}

\begin{axis}[xmin=-26,xmax=1,ymin=0,ymax=14.5,line width=1, font=\normalsize,
legend entries={upper bound, \textcolor{darkblue}{NFDM}, \textcolor{darkred}{WDM}},
legend style={at={(0.45,0.97)}, draw=none, fill=none, legend cell align=left},
xlabel={$\mathcal P$ [dBm]},
ylabel={AIR [bits/2D]},
y label style={at={(0.05,0.5)}},
xtick={-25, -20, -15, -10, -5, 0},
grid,
fill=none,
axis x line*=bottom,
width=0.45\textwidth
]

\addplot[ name path=lowerbound, color=black, densely dashed, mark size=1pt, line width=1.3] table {\awgnrate};

\addplot[ name path=lowerbound, color=darkblue, smooth, mark=*, mark size=1pt, line width=1.3] table {\nfdmrate};

\addplot[color=darkred, smooth, mark=triangle*, line width=1.3] table {\wdmrate};

\end{axis}

\begin{axis}[xmin=0,xmax=35,ymin=0,ymax=14.5,line width=1, font=\normalsize,
xlabel={SNR [dB]},
xlabel style={at={(0.5,1.3)}},
axis x line*= top,
xtick={1.27, 7.67, 14.2, 20.6, 27.2, 33.8},
x tick label style={yshift=2pt},
xtick style={draw=none},
width=0.45\textwidth
]

\addplot[color=white, draw=none] table [x expr=\thisrowno{0}+31, y expr=\thisrowno{1}]{\wdmrate};

\end{axis}

\end{tikzpicture}
\caption{AIRs of NFDM and WDM, and the capacity upper bound.}
\label{fig:nfdm-wdm-rates}
\end{figure}
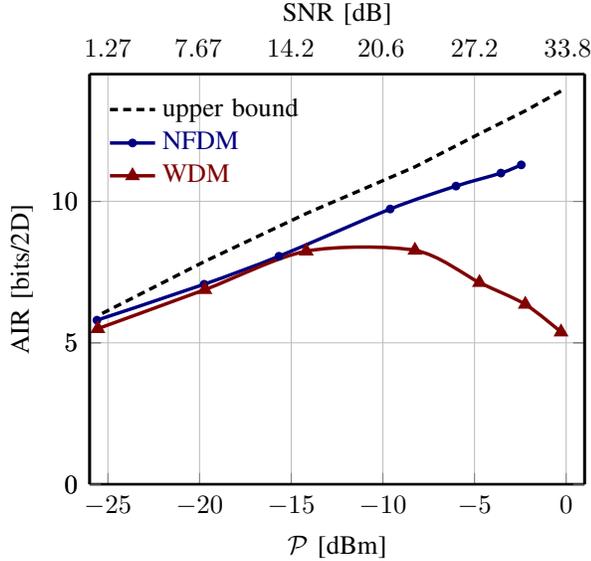

Fig.~\ref{fig:clouds} (b) shows the received symbols corresponding to four transmitted
symbols $s_0^0=0.7$, $s_0^0=0.98$, 
$s_0^0=1.2624$, $s_0^0=1.6$ in NFDM. The size of the `clouds'' does not increase notably 
as $|s_0^0|$ is increased. The corresponding constellation in WDM at the same power is presented in 
Fig.~\ref{fig:clouds}(c), showing received symbols for four transmitted symbols $s_0^0=0.7527$, $s_0^0=0.9575$, 
$s_0^0=1.185$, $s_0^0=1.458$. The WDM clouds in Fig.~\ref{fig:clouds}(c) are  
bigger than the NFDM clouds in Fig.~\ref{fig:clouds}(b). Note that in WDM, there is a 
rotation of symbols, even after back-propagation. This rotation, which is about $\gamma \const L\const P$ ($\gamma$ being
the nonlinearity coefficient), is due to the cross-phase modulation; see \cite[Eq. 21]{yousefi2014psd}.

Fig.~\ref{fig:entropy}(a)  shows that the conditional entropy in WDM increases with the input power, while
it is nearly constant in NFDM. Fig.~\ref{fig:entropy} (b) shows that the conditional probability
distribution $p(\hat s_0| s_0)\eqdef p_{\hat{S}_0|S_0}(\hat{s}_0| s_0)$ is shifted with $|s_0|$. 
Together these figures indicate that the channel in
the nonlinear Fourier domain is approximately an AWGN channel, for the signal and system parameters that we considered here.

\subsubsection{Spectral Efficiency}
\label{sec:wdm-nfdm-SEs}

Let $T(z)$ and $W(z)$ be the approximate time duration and
bandwidth of the signal at distance $z$.  
There are $T(0)W(0)$ complex DoFs in this time duration and bandwidth at $z=0$ ($WT\gg 1$). 
Among these, $N_uN_s$ complex DoFs are modulated in the input signal \eqref{eq:wdm-signal}. 
Define the \emph{modulation efficiency} $\eta$, $0 \leq \eta \leq 1$, as:
\begin{IEEEeqnarray*}{rCl}
\eta \eqdef \frac{N_uN_s}{\E T(0) \times\max\limits_{z} W(z)}.
\end{IEEEeqnarray*}
The higher is $\eta$, the more efficient is modulation. 
The SE $\rho$, in bits/s/Hz, is expressed in terms of the AIR as
\begin{IEEEeqnarray*}{rCl}
\rho =\eta R.
\end{IEEEeqnarray*}

We compare the SEs for one value of the average input power $\const P=-0.33$ dBm, and 
the parameters in Table~\ref{tab:params}. We obtain the AIRs: 
\begin{IEEEeqnarray*}{rCl}
R_{\textnormal{WDM}}=5.26,\quad  \quad 
R_{\textnormal{NFDM}}=10.5,\quad \textnormal{bits/2D}. 
\end{IEEEeqnarray*}
As a reference, the upper bound on the channel capacity at this power is $13.78$ bits/2D.
The modulation and spectral efficiencies are:
\begin{IEEEeqnarray*}{rClrCl}
\eta_{\textnormal{WDM}}&=&0.131, \quad  \rho_{\textnormal{WDM}} &=& 0.69 \quad \textnormal{bits/s/Hz},  \\
\eta_{\textnormal{NFDM}} &=& 0.147, \quad \rho_{\textnormal{NFDM}}&=& 1.54 \quad \textnormal{bits/s/Hz}.
\end{IEEEeqnarray*}
The gain in the SE is 2.23. 

We point out that the modulation efficiencies are small because $N_s$ is small. As a result, 
the above SEs are far below the maximum achievable SEs  corresponding to $N_s\rightarrow\infty$. 
For WDM, it can be proved analytically that as $N_s\rightarrow\infty$, $\eta\rightarrow 1$ and 
$\rho_{\textnormal{WDM}} \rightarrow R_{\textnormal{WDM}}$.  Accordingly, we anticipate that as 
$N_s\rightarrow\infty$,  
\begin{IEEEeqnarray*}{rCl}
\rho_{\textnormal{WDM}} &\rightarrow& 5.26 \quad \textnormal{bits/s/Hz},  \\
\rho_{\textnormal{NFDM}}&\rightarrow& 10.5 \quad \textnormal{bits/s/Hz}.
\end{IEEEeqnarray*}

\subsubsection{Limitations of the Results and Future Work}
We close this section by pointing out some of the limitations of our results and chart directions for research.

\paragraph{Non-ideal models} 
A looming weakness of NFDM is that it applies only to integrable models of the optical 
fiber. An example is 
the lossless noise-free NLS equation with second-order dispersion and 
cubic nonlinearity.  Following the methodology established in this paper, AIRs of NFDM and WDM 
in the presence of perturbations --- such as loss, higher-order dispersion and polarization 
effects ---  were recently studied in \cite{yangzhang2018jlt}. It is shown that uncompensated perturbations 
reduce the AIRs of both schemes. A conclusive comparison of the AIRs with perturbations compensation 
is still open research. 

\paragraph{AIRs as $N_s\rightarrow\infty$}

In this paper, we considered $N_s=1$, whereas in practice $N_s$ can be over several hundred.   
The asymptotic capacity $C(\const P, n)$ of a discrete-time model of the optical fiber as a function of the 
input power $\const P$ and the number of DoFs 
$n=N_sN_u$ is established in \cite{yousefi2016cap}. The capacity formula \cite[Thm.~1]{yousefi2016cap} shows 
that, for fixed $\const P$, as $N_s$ is increased the capacity is decreased. 
As explained in \cite{yousefi2016cap}, the reason is that the signal-noise interaction 
grows with $N_s$. Therefore, as $N_s$ is increased, the AIR of NFDM and WDM may decrease (with
$N_s$, not $\const P$).
The conclusion that the 
NFDM outperforms WDM for $N_s=1$ is yet to be examined for $N_s>1$.
In the system considered in this paper, the signal-noise interaction
is expected to factor in similarly in both schemes. 

Note that as $N_s\rightarrow\infty$, the peak-to-average power ratio (PAPR)  of $\hat q(\lambda,0)$
often increases. This leads to numerical error as $\hat q(\lambda,0)\rightarrow 1$. To enable simulations 
with $N_s>1$, i) methods for reducing the PAPR in OFDM
can be applied; ii) data can be directly
modulated in \eqref{eq:wdm-signal}; or iii) transformations other than \eqref{eq:X-qhat} can be considered.

\paragraph{Realistic parameters}

The simulations in this paper are performed with $B=60$ GHz, $N_u=15$, and $N_s=1$, 
which do not correspond to practical systems. Simulations with realistic values for these parameters require 
more computational resources, time, and possibly refined algorithms. The value of the dispersion parameter in the paper 
is $D=-17$ ps/(nm-km); in realistic systems $D$ varies widely depending on the type of the
fiber; for instance $D=-4.6$  ps/(nm-km) in a fiber used in submarine applications. The value of $D$ 
should not play a significant role in NFDM.

\begin{figure*}[t]
\centerline{
\begin{tabular}{c@{~~~~}c@{~~~~}c}
\includegraphics[width=0.275\textwidth]{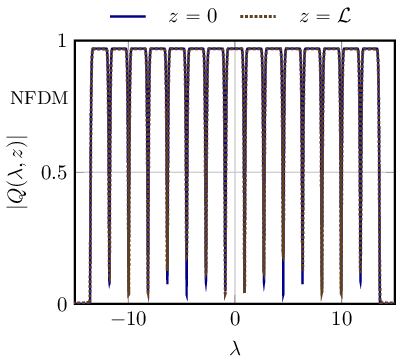}
&
\includegraphics[width=0.275\textwidth]{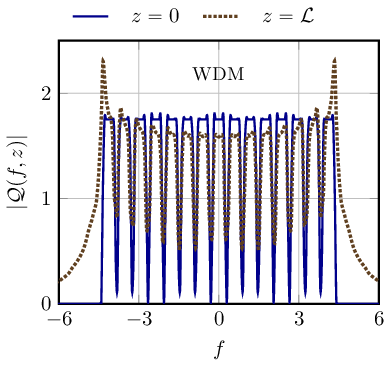}
&
\includegraphics[width=0.275\textwidth]{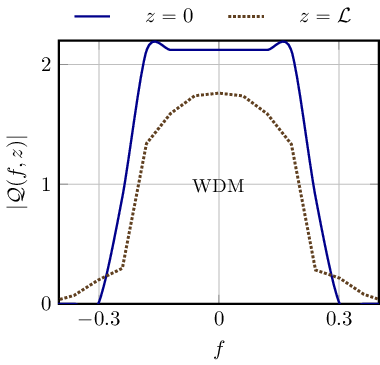}
\\
~~~~~~(a) & ~~~~~(b) & ~~~~~(c)
\end{tabular}
}
\caption{Interference in WDM and NFDM in the absence of noise, for signals with the same power, bandwidth and time duration. 
(a) No interference in NFDM. Interference in WDM: (b) before equalization, and
(c) after equalization.}
\label{fig:demonstrations}
\end{figure*}

\begin{figure}[t]
\centerline{
\begin{tabular}{cc}
\includegraphics[width=0.225\textwidth]{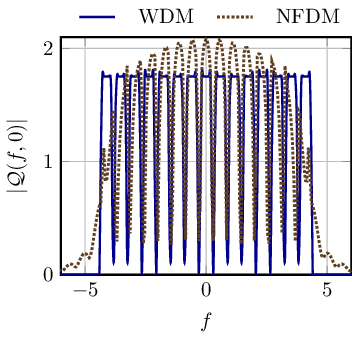}
&
\includegraphics[width=0.235\textwidth]{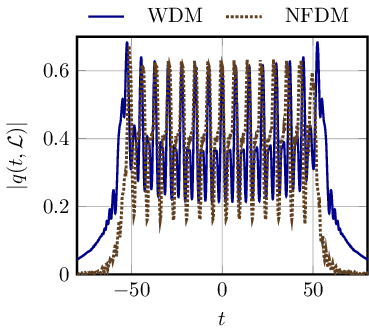}
\\
~~~~~~(a) & ~~~~~~~(b)
\end{tabular}
}
\caption{ 
WMD and NFDM signals have the same 
(a) bandwidth $B=60$ GHz (that is maximum at $z=0$), and 
(b) time duration 1.5 ns (that is maximum at $z=2000$ km).
The power is  $\const P=0$ dBm. 
} 
\label{fig:NFDM-WDM-comparisons}
\end{figure}

\begin{figure}[t]
\centerline{
\begin{tabular}{cc}
\includegraphics[width=0.24\textwidth]{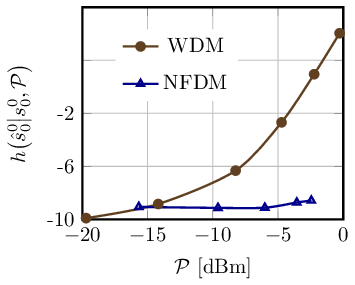}
&
\includegraphics[width=0.24\textwidth]{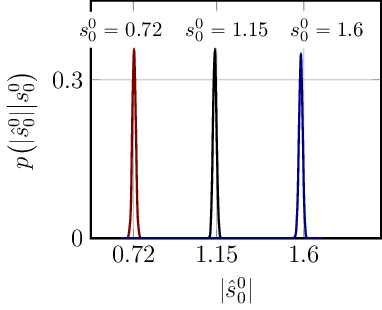}
\\
~~~~~ (a) & ~~~~~(b)
\end{tabular}
}
\caption{ (a) Compared with WDM, the conditional differential entropy in NFDM is nearly constant as $\const P$ is
  increased ($s_0^0=0.2$). 
  (b) The conditional probability distribution $p\bigl(|\hat s_0^0| \bigl| s_0^0 \bigr)$ in NFDM, for
  three values of $s_0^0$.}
\label{fig:entropy}
\end{figure}

\begin{figure*}[t]
\centering
\begin{tabular}{c@{~~~~}c@{~~~~}c}
\includegraphics[width=0.3\textwidth]{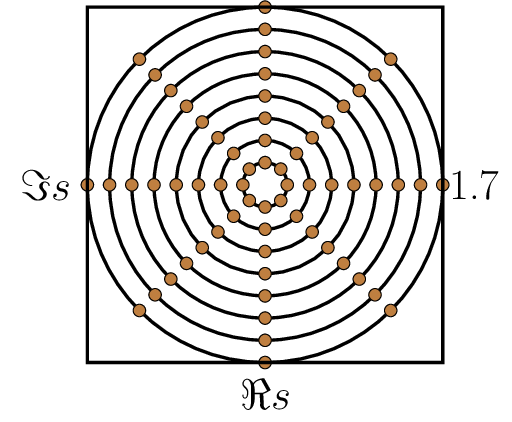}
&
\includegraphics[width=0.3\textwidth]{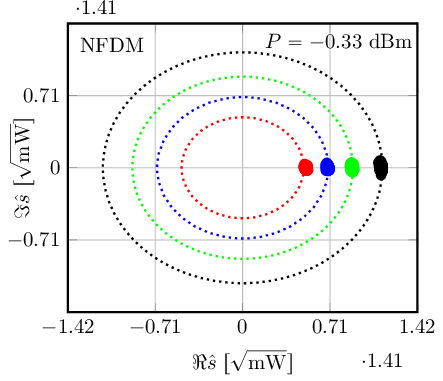}
&
\includegraphics[width=0.3\textwidth]{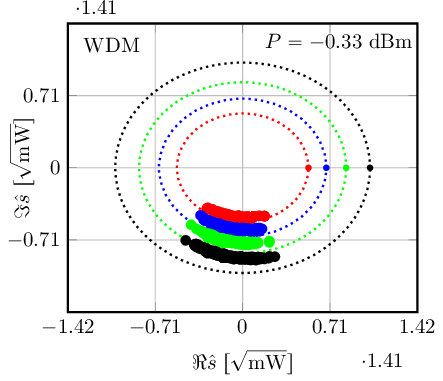}
\\
~(a) & (b) & (c)
\end{tabular}
\caption{(a) Constellation for $s_\ell^k$ in the $U$ domain. 
Received symbols for four transmitted symbols in (b) NFDM, and
(c) WDM.}
\label{fig:clouds}
\end{figure*}


\section{Conclusion}
\label{sec:conclusions}

The paper shows that the NFDM AIR is greater than the WDM
AIR for a given  power and bandwidth, in an integrable model
of the optical fiber in the defocusing regime, and in a representative system 
with one symbol per user. While the paper serves as a good starting  point, 
more research is needed in comparing the linear and nonlinear multiplexing.


\section*{Acknowledgments}

The research was done when the authors were at TU Munich in Germany. 
Financial support of the Institute for Advanced Study at TU Munich, 
funded by the German Excellence
Initiative, as well as Alexander 
von Humboldt Foundation, funded by the German
Federal Ministry of Education and Research, is acknowledged.
Mansoor Yousefi benefited from discussions with Frank Kschischang and Gerhard Kramer.


\appendices

\section{Proof of \eqref{eq:qk-from-vk+1-ALL} and \eqref{eq:Q[k-1]}}

\label{sec:A-B}

The first few iterations in the scaled AL scheme are:
\begin{IEEEeqnarray*}{rCl}
\begin{pmatrix}
A[0,z]\\
B[0,z]
\end{pmatrix}
&=&\begin{pmatrix}
1\\
0
\end{pmatrix}
,\quad
\begin{pmatrix}
A[1,z]\\
B[1,z]
\end{pmatrix}
=
\bar c_0
\begin{pmatrix}
1\\
sQ^*_0
\end{pmatrix}
,\\
\begin{pmatrix}
A[2,z]\\
B[2,z]
\end{pmatrix}
&=&\bar c_1
\begin{pmatrix}
1+sQ^*_0Q_1z^{-1}\\
sQ^*_1+sQ^*_0z^{-1}
\end{pmatrix},
\\
\begin{pmatrix}
A[3,z]\\
B[3,z]
\end{pmatrix}
&=& \bar c_2
\begin{pmatrix}
1+s(Q^*_0Q_0+Q^*_1Q_2)z^{-1}+sQ^*_0Q_2z^{-2}\\
sQ^*_2-s(Q^*_0Q_1Q^*_2-Q^*_1)z^{-1}+sQ^*_0z^{-2}
\end{pmatrix}
,
\end{IEEEeqnarray*}
where $Q_k\eqdef Q[k]$ and $\bar c_k\eqdef\prod_{i=0}^{k} c_i$. 

For $k=0,1,2$, the coefficients of $A[k+1,z]$ and $B[k+1,z]$ with the smallest power of $z^{-1}$ 
are 
\begin{IEEEeqnarray}{rCl}
A_0[k+1]=\bar c_k,\quad B_0[k+1]=s\bar c_k Q^*[k].
\label{eq:A0-B0-proof}
\end{IEEEeqnarray}
By induction, \eqref{eq:A0-B0-proof} holds for all $k\geq 0$. Thus, 
$Q^*[k]=sB_0[k+1]/A_0[k+1]$, which is \eqref{eq:qk-from-vk+1-ALL}.

Similarly, for $k=1,2$, the coefficients of $A[k+1,z]$ and $B[k+1,z]$ 
with the highest power of $z^{-1}$ are 
\begin{IEEEeqnarray*}{rCl}
A_{k}[k+1]=s\bar c_k Q_0^*Q_k,\quad B_k[k+1]=s\bar c_kQ^*_0. 
\end{IEEEeqnarray*}
By induction, the last equation holds for all $k\geq 1$. Thus, 
$Q[k]=A_{k}[k+1]/B_{k}[k+1]$, which is \eqref{eq:Q[k-1]}.


\section{Kramers-Kronig Relations}
\label{sec:kramers-kronig}
The real and imaginary parts of an analytic function are
related via the Kramers-Kronig relations.

\begin{lemma}[Kramers-Kronig Relations]
\label{lemm:kramers-kronig}
Let $f(z)\eqdef u(x,y)+jv(x,y)$ be a function of a complex variable
$z\eqdef x+jy\in\Complex^+$, where $u$ and $v$ are real-valued functions. Suppose that 
$f(z)$ is analytic in $\Complex^+$ and $f(z)$ decays as $1/z$ as 
$|z|\rightarrow\infty$. Then 
\begin{IEEEeqnarray}{rCl}
u(x,y)=-\mathcal H_x(v(x,y)),\quad v(x,y)=\mathcal H_x(u(x,y)),
\label{eq:hilbert-transform-eqs}
\end{IEEEeqnarray}
where $\mathcal H_x$ denotes the Hilbert transform with respect to $x$: 
\begin{IEEEeqnarray*}{rCl}
\mathcal H_x(f(x,y))&\eqdef&\frac{1}{\pi}\textnormal{p.v.}\int\limits_{-\infty}^\infty \frac{f(x',y)}{x-x'}\der x'\\
&=&\frac{1}{\pi x}\convolution f(x,y),
\end{IEEEeqnarray*}
where $\convolution$ is convolution with respect to $x$, and p.v. is the
principal value.
\end{lemma}

\begin{proof}
The result follows immediately from the Cauchy's integral formula, or 
the Sokhotski-Plemelj formula \cite[Lem.~8]{yousefi2012nft1}.

\end{proof}

\begin{lemma}
\label{cor:kramers-kronig}
Let $q(t)\in L^1(\Reals)$ and $a(\lambda)$  be the corresponding 
nonlinear Fourier coefficient in the defocusing regime. Then
\begin{IEEEeqnarray}{rCl}
\angle(a(\lambda))=\mathcal H(\log|a(\lambda)|),
\label{eq:|a|-via-angle(a)}
\end{IEEEeqnarray}
for all $\lambda\in\Reals$ for which $\angle(a(\lambda))\in(-\pi,\pi)$. 

\qed
\end{lemma}

\begin{proof}

If $q(t)\in L^1(\Reals)$ and $s=1$,  $a(\lambda)$ can be analytically 
extended to $\Complex^+$ and is 
continuous in $\bar \Complex^+\eqdef\bigl\{\lambda\in\Complex: \Im(\lambda)\geq 0\bigr\}$
\cite[Lemma~2.1]{ablowitz2003dcn}, \cite[Lem.~4]{yousefi2012nft1}. 
Consider the open region $\mathcal{D}\eqdef\{\lambda\in\Complex: 0<\Im(\lambda)<\epsilon\}$ with 
$\epsilon\rightarrow 0^+$. 
Note $a(\lambda)$ is analytic on $\mathcal{D}$ and continuous on its closure $\bar{\mathcal D}$. 
From the unimodularity condition \eqref{eq:C-unimodularity} with $s=1$
\begin{IEEEeqnarray*}{rCl}
  |a|^2&=&1+|b|^2\\
 &\geq & 1,\quad \lambda\in\Reals,
\end{IEEEeqnarray*}
thus $|a(\lambda)|\neq 0$ for $\lambda\in\Reals$. Because $a(\lambda)$ is continuous on $\bar{\mathcal D}$,
we obtain $a(\lambda)\neq 0$ on $\bar{\mathcal{D}}$ (recall $\epsilon\rightarrow 0$).
Taking the principal branch of the logarithm, function $\lambda\mapsto\log(\lambda)$ is continuous 
in $\Complex\backslash\{0\}$ and analytic on $\Complex\backslash\Realsnp$, where $\Realsnp$ 
is a branch cut of the logarithm. 
Since $a(\lambda)$ is analytic and non-zero on $\mathcal{D}$, and $\textnormal{range}(a(\lambda)) \subseteq\Complex\backslash\Realsnp$, 
$\log a(\lambda)$ is analytic on $\mathcal D$, and continuous on $\bar{\mathcal{D}}$.

We next show that $\log a(\lambda)$ satisfies the decay 
assumption in  Lemma~\ref{lemm:kramers-kronig}.
From \cite[Eq. 45]{yousefi2012nft1}, 
\begin{IEEEeqnarray*}{rCl}
  a(\lambda)=1-\frac{jE}{2}\lambda^{-1}+\mathcal O(\lambda^{-2}),\quad \textnormal{as}\quad |\lambda|\rightarrow\infty,
\end{IEEEeqnarray*}
where $E\eqdef\int\limits_{-\infty}^{\infty}|q(t)|^2\der t$.
Thus
\begin{IEEEeqnarray}{rCl}
|\log a(\lambda)|\rightarrow 0.5 E\lambda^{-1},\quad \textnormal{as},\quad |\lambda|\rightarrow\infty.
\label{eq:loga-inf}
\end{IEEEeqnarray}
Applying Lemma~\ref{lemm:kramers-kronig} on $\mathcal D$ gives
\begin{IEEEeqnarray}{rCl}
\angle (a(\lambda))=\mathcal H(\log|a(\lambda)|),\quad \lambda\in\mathcal D.
\label{eq:|a|-angle}
\end{IEEEeqnarray}

We show \eqref{eq:|a|-angle} holds for $\lambda\in\Reals$ using continuity. 
For $x\in\Reals$
\begin{IEEEeqnarray*}{rCl}
\angle(a(x))&=&
\angle\bigl(a(\lim\limits_{\epsilon\rightarrow 0^+}(x+j\epsilon))\bigr)
\\
&\overset{(a)}{=}&
\angle\bigl(
 \lim\limits_{\epsilon\rightarrow 0^+} a(x+j\epsilon)\bigr)
\\&\overset{(b)}{=}&\lim\limits_{\epsilon\rightarrow 0^+} \angle(a(x+j\epsilon))
\\
&\overset{(c)}{=}&
œ\lim\limits_{\epsilon\rightarrow 0^+}
\mathcal{H}\bigl(\log|a(x+j\epsilon)|\bigr)
\\
&\overset{(d)}{=}&
\mathcal{H}\Bigl(
\lim\limits_{\epsilon\rightarrow 0^+}\log|a(x+j\epsilon)|\Bigr)
\\
&\overset{(e)}{=}&
\mathcal{H}(\log|a(x)|).
\label{eq:|a|-via-angle(a)-2}
\end{IEEEeqnarray*}
Step $(a)$ follows because $a(\lambda)$ is continuous in $\bar{\mathcal{D}}$.
Step $(b)$ follows because $a(\lambda)$ is continuous and non-zero in $\bar{\mathcal{D}}$, 
and $\lambda\mapsto\log(\lambda)$ is
continuous in $\Complex\backslash\{0\}$, thus $\angle(a)\in(-\pi,\pi)$ is continuous 
in its principal branch.
Step $(c)$ follows from \eqref{eq:|a|-angle}. 
Step $(d)$ holds because $\mathcal H$ is a continuous operator on the
space of functions $\log|a(\lambda)|$, as we show below.
Step $(e)$ holds because $\log|a(\lambda)|$ is continuous on
$\bar{\mathcal{D}}$.

We show that the order of the limit and the integral can be exchanged in the following Hilbert transform: 
\begin{IEEEeqnarray*}{rCl}
\angle(a(x))=
\frac{1}{\pi}
\lim\limits_{\epsilon\rightarrow 0^+}
\textnormal{p.v.}
\int\limits_{-\infty}^{\infty}
\frac{\log|a(x'+j\epsilon)|}
{x-x'}\der x'.
\end{IEEEeqnarray*}
Let $\delta>0$ be sufficiently small and $K<\infty$ sufficiently large. We partition the integration range into the union of 
$I_0\eqdef\{ x' :  |x'-x|<\delta\}$, $I_1\eqdef\{ x': |x'|\leq K,\: x'\notin I_0\}$ and 
$I_2\eqdef\{ x': |x'|>K\}$. From \eqref{eq:loga-inf}, there exists $C_2<\infty$ 
such that
\begin{IEEEeqnarray}{rCl}
\left|
  \frac{\log|a(x'+j\epsilon)|}{x'-x}\right|<\frac{C_2}{|x'(x'-x)|},
\quad x'\in I_2.
\label{eq:loga-ub}
\end{IEEEeqnarray}
The right hand side in \eqref{eq:loga-ub} is independent of $\epsilon$
and integrable for $x'\in I_2$.
Likewise, because $a(\lambda)$ is continuous on $\bar{\mathcal D}$ and $I_1$ is compact,
from the extreme value theorem $|a(x'+j\epsilon)|<C_1$ for $x'\in I_1$ and some $C_1<\infty$, independently of $\epsilon$. Thus  
the integrand corresponding to $I_1$ is uniformly upper bounded by an absolutely integrable function. 
Consequently, from the dominated convergence theorem the order of the limit and the integral 
can be exchanged in integrals over $I_1$ and $I_2$.  Finally, since
$\log|a(\lambda)|$ is continuous,
the integral over $I_0$ evaluates to
$\log |a(x)|\times\textnormal{p.v.}\int_{-\delta}^{\delta}\frac{1}{x}\der
x=0$.

\end{proof} 

\begin{remark}
If $q(t)\in L^1(\Reals)$, $a(\lambda)$ is analytic on $\Complex^+$, but not necessarily on $\bar\Complex^+$, unless 
$q(t)$  vanishes exponentially as $t\rightarrow \pm\infty$. 
For such signals, $\log a(\lambda)$ is analytic on $\bar\Complex^+$, and the proof of Lemma~\ref{cor:kramers-kronig} 
simplifies.

\qed
\end{remark}




\end{document}